\documentclass[10pt,a4paper,english]{article}
\usepackage{amsmath,amscd,amsfonts,eucal,latexsym,amssymb,mathrsfs}
\usepackage{amsxtra, amsthm, calc, bbm}
\newlength{\papersidemargin}
\newlength{\papertopmargin}
\newlength{\maxcarwidth}
\theoremstyle{definition}

\theoremstyle{plain}

\theoremstyle{remark}
\newtheorem*{Remark}{Remark}

\theoremstyle{definition}
\newtheorem{definition}{Definition}[section]
\theoremstyle{plain}
\newtheorem{theorem}[definition]{Theorem}
\newtheorem{proposition}[definition]{Proposition}
\newtheorem{lemma}[definition]{Lemma}

\numberwithin{equation}{section}
\newcounter{remcount}

\newenvironment{remlist}{\begin{list}{(\roman{remcount})}{\usecounter{remcount}\setlength{\leftmargin}{0 cm}\setlength{\rightmargin}{0 cm}\setlength{\topsep}{0 cm}\setlength{\itemsep}{0 cm}\setlength{\parsep}{\parskip}\setlength{\labelwidth}{1 cm}\setlength{\labelsep}{0.5 em}\setlength{\itemindent}{1 cm + 0.5 em}}}{\end{list}}
\newcounter{propcount}

\newlength{\maxlabelwidth}
\newenvironment{proplist}[2][1]{\begin{list}{(\roman{propcount})}{\usecounter{propcount}\setcounter{propcount}{#2}\settowidth{\maxlabelwidth}{\textit{(\roman{propcount})}}\setcounter{propcount}{#1-1}\setlength{\leftmargin}{\maxlabelwidth+ 
0.5 em}\setlength{\rightmargin}{0 cm}\setlength{\topsep}{0 
cm}\setlength{\itemsep}{0 
cm}\setlength{\parsep}{\parskip}\setlength{\labelwidth}{\maxlabelwidth}\setlength{\labelsep}{0.5 
em}\setlength{\itemindent}{0 cm}}}{\end{list}}
\def\bC{{\mathbb C}}

\def\bN{{\mathbb N}}

\def\bR{{\mathbb R}}

\def\a{\alpha}
\def\b{\beta}
\def\g{\gamma}        
\def\d{\delta}        \def\D{\Delta}
\def\eps{\varepsilon}

\def\l{\lambda}       
\def\m{\mu}

\def\r{\rho}
\def\s{\sigma}

\def\t{\tau}

\def\o{\omega}        \def\O{\Omega}
\newcommand{\sC}{\mathscr{C}}
\newcommand{\sD}{\mathscr{D}}

\newcommand{\sF}{\mathscr{F}}

\newcommand{\sH}{\mathscr{H}}

\newcommand{\sN}{\mathscr{N}}

\newcommand{\sS}{\mathscr{S}}

\newcommand{\supp}{{\textup{supp}\,}}
\newcommand{\ad}{{\textup{Ad}\,}}

\newcommand{\Id}{\mathbbm{1}} 

\newcommand{\de}{\partial}

\newcommand{\vntensor}{\bar{\otimes}}

\newcommand{\abs}[1]{\lvert#1\rvert}
\newcommand{\bigabs}[1]{\big\lvert#1\big\rvert}

\newcommand{\biggabs}[1]{\bigg\lvert#1\bigg\rvert}
\newcommand{\norm}[1]{\| #1 \|}

\newcommand{\scalar}[2]{\langle #1  ,  #2\rangle}

\newcommand{\bp}{\boldsymbol{p}}
\newcommand{\bq}{\boldsymbol{q}}
\newcommand{\bk}{\boldsymbol{k}}
\newcommand{\bx}{\boldsymbol{x}}

\newcommand{\tsH}{\tilde{\sH}}
\newcommand{\tsF}{\tilde{\sF}}
\newcommand{\omp}{\o_m(\bp)}
\newcommand{\omq}{\o_m(\bq)}

\newcommand{\wick}[1]{:\mspace{-3mu}#1\mspace{-7mu}:\mspace{-3mu}}

\newcommand{\WlOrd}{W_{\l O_r,\l O_{r+\d}}}
\newcommand{\XilO}{\Xi_{\l O_r,\l O_{r+\d}}}
\newcommand{\XiO}{\Xi_{O_r,O_{r+\d}}}

\newcommand{\inst}[1]{$^\textrm{#1}$ }
\newcommand{\email}[1]{e-mail: #1}

\begin{document}
\title{From global symmetries to local currents: \\
the free (scalar) case in 4 dimensions
\thanks{%
Work supported by MIUR, GNAMPA-INDAM, the SNS and the EU network
``Quantum Spaces -- Non Commutative Geometry'' (HPRN-CT-2002-00280).}%
}
\author{Gerardo Morsella\inst{1}
\and Luca Tomassini\inst{2}}

\date{
\parbox[t]{0.9\textwidth}{\footnotesize{%
\begin{enumerate}
\renewcommand{\theenumi}{\arabic{enumi}}
\renewcommand{\labelenumi}{\theenumi}
\item Dipartimento di Matematica, Universit\`a di Roma ``Tor Vergata'', Via della Ricerca Scientifica, I-00133 Roma, Italy, \email{morsella@mat.uniroma2.it} 
\item Dipartimento di Matematica, Universit\`a di Roma ``Tor Vergata'', Via della Ricerca Scientifica, I-00133 Roma, Italy, \email{tomassin@mat.uniroma2.it} 
\end{enumerate}
}}
\\
\vspace{\baselineskip}
November 2, 2009}

\maketitle

\begin{abstract}
Within the framework of algebraic quantum field theory, we propose a new method of constructing local generators of (global) gauge symmetries in field theoretic models, starting from the existence of unitary operators implementing locally the flip automorphism on the doubled theory. We show, in the simple example of the internal symmetries of a multiplet of free scalar fields, that through the pointlike limit of such local generators the conserved Wightman currents associated with the symmetries are recovered.
\end{abstract}

\section{Introduction}\label{sec:intro}
One of the most important features of field theoretic models is the existence of local conserved currents corresponding to space-time and internal (gauge) symmetries. While in the framework of classical lagrangian field theory a clarification of this issue comes from Noether's theorem (which provides an explicit formula for the conserved current associated to any \emph{continuous} symmetry of the Lagrangian itself), it is well known that in the quantum case several drawbacks contribute to make the situation more confuse. For example, symmetries which are present at the classical level can disappear upon quantization due to renormalization effects.

In \cite{Doplicher:1983} a different approach to the problem was outlined in the context of algebraic quantum field theory. It consisted of two main steps: 1) given double cones $O$, $\hat O$  with bases $B$, $\hat B$ in the time-zero plane centered at the origin and such that $\bar O\subset\hat{O}$, start from generators $Q$ of global space-time or gauge transformations and construct local ones, i.e.\ operators $J^Q_{O,\hat{O}}$ generating the correct symmetry on the field algebra $\sF(O)$ and localized in $\hat{O}$ (i.e.\ affiliated to $\sF(\hat{O})$); 2) these local generators should play the role of integrals of (time components of) Wightman currents over $B$ with a smooth cut-off in $\hat B$ and possibly some smearing in time, so that one is led to conjecture that \begin{equation}\label{intro}
 \frac{1}{\l^3}\int_{\bR^4} f(x)\a_x\big(J^Q_{\l O,\l\hat O}\big)dx \rightarrow c j_0^Q(f)
\end{equation}
holds, in a suitable sense, as $\l\to 0$. Here $\a$ denotes space-time translations, $j_0^Q(x)$ the sought for Wightman current, $f \in \sS(\bR^4)$ any test function and $c$ a constant which (in view of the above interpretation of $J^Q_{O,\hat O}$) would be expected to satisfy
\begin{equation}\label{eq:vol}
\text{vol}(B) \leq c \leq \text{vol}(\hat B).
\end{equation}
It is important to note that there is a large ambiguity in the choice of the local generators: since their action in $O^{\prime}\cap\hat{O}$ is not fixed we are free to add perturbations in $\sF(O^{\prime}\cap\hat{O})$. Thus, the limit \eqref{intro} is not to be expected to converge in full generality, but we can still hope that a ``canonical'' choice or construction of the local generators might solve the problem (see below).

The first problem above was completely solved in \cite{Doplicher:1982} for the case of abelian gauge transformation groups, while in \cite{Doplicher:1983} and \cite{Buchholz:1986} the general case (including discrete and space-time symmetries and supersymmetries) was treated. The final result was that in physically reasonable theories what was called by the authors a canonical local unitary implementation of global symmetries exists and if a part of them actually constitutes a Lie group the corresponding canonical local generators provide a local representation of the associated Lie (current) algebras. A key assumption was identified in the so-called split property (for double cones), which holds in theories with a realistic thermodynamic behaviour \cite{Buchholz:1986dy}. It expresses a strong form of statistical independence between the regions $O$ and $\hat{O}^{\prime}$ and is equivalent to the existence of normal product states $\phi$ on $\sF(O) \vee \sF(\hat O)'$ such that $\phi(AB)=\omega(A)\omega(B)$ ($\omega$ being the vacuum state) for $A\in \sF(O)$ and $B\in \sF(\hat{O})^{\prime}$ \cite{Buchholz1974}.

However, the above mentioned construction crucially depends on such a highly elusive object as the unique vector representative of the state $\phi$ in the (natural) cone
\begin{equation}
P^{\natural}_{\Omega}(\sF(O)\vee\sF(\hat{O})^{\prime})=\overline{\Delta^{1/4}(\sF(O)\vee\sF(\hat{O}^{\prime}))_+\Omega}
\end{equation}
(see \cite{Stratila:1981a}), where $\Omega$ indicates the vacuum vector and $\Delta$ the modular operator of the pair $(\sF(O) \vee \sF(\hat O)',\Omega)$, so that finding an explicit expression of the local generators appears as an almost hopeless task. This makes it extremely hard to proceed to the above mentioned second step, i.e.\ the determination of the current fields themselves. Notwithstanding this, the reconstruction of the energy momentum tensor of a certain (optimal) class of  2-dimensional conformal models was carried out in \cite{Carpi1999}, while partial results for the U(1)-current in the free massless four dimensional case were obtained in \cite{Tomassini1999}, showing that for the local generators of~\cite{Buchholz:1986} the drawbacks briefly discussed after equation \eqref{intro} might be less severe. However, in both cases the existence of a unitary implementation of dilations was crucial for handling the limit $\l\to0$.

In what follows we restrict our attention to the case of continuous symmetries and propose a new method for obtaining local generators based on the existence of local unitary implementations of the flip automorphism, a requirement actually equivalent, under standard assumptions, to the split property \cite{D'Antoni:1983a}. This method turns out to be particularly suited for carrying out step 2) above, at least in the free field case. 

To be more specific, we consider a quantum field theory defined by a net $O \to \sF(O)$ of von Neumann algebras on open double cones in Minkowski 4-dimensional spacetime acting irreducibily on a Hilbert space $\sH$ with scalar product $\scalar{\cdot}{\cdot}$ satisfying the following standard assumptions:
\begin{remlist}
\item there is a unitary strongly continuous representation $V$ on $\sH$ of a compact Lie group $G$, which acts locally on $\sF$
\begin{equation*}
V(g) \sF(O) V(g)^* = \sF(O),\qquad g \in G,
\end{equation*}
and we set $\b_g := \ad V(g)$;
\item (\emph{split property}) for each pair of double cones $O_1 \subset\subset O_2$ (i.e.\ $\bar O_1 \subset O_2$) there exists a type I factor $\sN$ such that
\begin{equation*}
\sF(O_1)\subset\sN\subset\sF(O_2).
\end{equation*}
\end{remlist}

To such a theory, we associate the doubled theory $O \to \tsF(O) := \sF(O) \vntensor \sF(O)$, with the corresponding unitary represention of $G$ given by $\tilde V(g) := V(g)\otimes V(g)$.
In this situation, it is well known that for each pair of double cones $O_1 \subset\subset O_2$ there exists a local implementation of the flip automorphism of $\tsF(O_1)$, i.e.\ a unitary operator $W_{O_1,O_2}\in \tsF(O_2)$ such that
\begin{equation}\label{eq:localimplemW}
W_{O_1,O_2} F_1 \otimes F_2 W_{O_1,O_2}^* = F_2 \otimes F_1, \qquad F_1,F_2 \in \sF(O_1).
\end{equation}

Assume now, for the argument's sake, that there is a one parameter subgroup $\theta \in \bR \to g_\theta \in G$ of $G$, such that the generator $Q$ of the corresponding unitary group $\theta \to V(g_\theta)$ is a \emph{bounded} operator on $\sH$. Considering the conditional expectation (Fubini mapping) $E_\Phi : B(\sH)\vntensor B(\sH) \to B(\sH)$ defined by
\begin{equation*}
E_\Phi(A_1\otimes A_2) = \langle \Phi,A_2\Phi\rangle A_1, \qquad A_1,A_2 \in B(\sH),
\end{equation*}
where $\Phi \in \sH$ is such that $\|\Phi\|=1$,
we can define the operator
\begin{equation}\label{eq:J}
J^Q_{O_1,O_2}:=\Xi_{O_1,O_2}^{\Phi}(Q) := E_\Phi(W_{O_1,O_2} (\Id \otimes Q)W_{O_1,O_2}^*), 
\end{equation} 
and it is then easy to see that such operator gives a \emph{local implementation of the infinitesimal symmetry generated by $Q$} in the following natural sense:
\begin{equation}\label{eq:localimplemJ}
J^Q_{O_1,O_2} \in \sF(O_2), \qquad [J^Q_{O_1,O_2}, F] = [Q,F], \qquad \forall\, F \in \sF(O_1).
\end{equation}
We also note that for this last equation to hold, it is sufficient that $W_{O_1,O_2}$ is only a \emph{semi-local} implementation of the flip, i.e.\ a unitary in $\sF(O_2)\vntensor B(\sH)$ for which~\eqref{eq:localimplemW} holds.

The assumption of boundedness for $Q$ is of course very strong, and it is not expected to be satisfied in physically interesting models. In the unbounded case it is however possible, in the slightly more restrictive setting of~\cite{Doplicher:1983, Buchholz:1986}, to make sense of equations~\eqref{eq:J}, \eqref{eq:localimplemJ} producing a self-adjoint operator $J^Q_{O_1,O_2}$ affiliated to $\sF(O_2)$ and implementing the commutator with $Q$ on a suitable dense subalgebra of $\sF(O_1)$.  More explicitly, assume that the triple $\Lambda =(\sF(O_1),\sF(O_2),\O)$ is a standard split W$^*$-inclusion in the sense of~\cite{Doplicher:1984a} and consider the unitary standard implementation $U_\Lambda : \sH \to \sH\otimes\sH$ of the isomorphism $\eta : F_1F_2' \in \sF(O_1)\vee\sF(O_2)' \to F_1\otimes F_2'\in\sF(O_1)\vntensor\sF(O_2)'$. This was used in~\cite{Buchholz:1986} to define the universal localizing map $\psi_{\Lambda}: B(\sH) \to B(\sH)$,
\begin{equation*}
\psi_{\Lambda}(T) = U_{\Lambda}^*(T\otimes\Id )U_{\Lambda}, \qquad T\in B(\sH),
\end{equation*}
where the standard type-I factor $\sN_{\Lambda}=\psi_\Lambda(B(\sH))$ satisfies $\sF(O_1)\subset \sN_{\Lambda}\subset \sF(O_2)$. For the commutant standard inclusion  $\Lambda^{\prime}=(\sF(O_2)^{\prime},\sF(O_1)^{\prime},\O)$~\cite{Doplicher:1984a}, one has $\psi_{\Lambda'}(T) = U_\Lambda^*(\Id\otimes T)U_\Lambda$.

For any unitarily equivalent triple $\Lambda_{0}=(V_{0}\sF(O_1)V_{0}^*,V_{0}\sF(O_2)V_{0}^*,V_{0}\O)$, one finds $U_{\Lambda_0}\cdot V_0 = V_0\otimes V_0 \cdot U_{\Lambda}$. Notice that in the case of gauge transformations $\Lambda =\Lambda_0$ and so 
\begin{equation}\label{eq:cov}
U_{\Lambda} V(g) = V(g)\otimes V(g) \cdot U_{\Lambda}.
\end{equation}

It is then straightforward to verify that, with $Z_{1,3}$ the unitary interchanging the first and third factors in $\sH\otimes\sH\otimes\sH\otimes\sH$, the operator
\begin{equation*}
W_\Lambda = (U_\Lambda^*\otimes U_\Lambda^*)Z_{1,3}(U_\Lambda\otimes U_\Lambda)
\end{equation*}
is a local implementation of the flip. Setting $g = g_\theta$ in~\eqref{eq:cov} and differentiating with respect to $\theta$, a simple computation shows that
\begin{equation*}
W_\Lambda (\Id\otimes Q) W_\Lambda^*= \psi_{\Lambda}(Q)\otimes\Id + \Id\otimes \psi_{\Lambda^{\prime}}(Q)=J_\Lambda^Q\otimes \Id + \Id \otimes J_{\Lambda^{\prime}}^Q,
\end{equation*}
where $J_\Lambda^Q, J_{\Lambda'}^Q$ are the canonical local implementations of~\cite{Doplicher:1983, Buchholz:1986}, which of course satisfy~\eqref{eq:localimplemJ}. Choosing now $\Phi = U_\Lambda^*(\Omega\otimes\Omega)$, we see that $\Xi_{O_1,O_2}^{\Phi}(Q)= J_\Lambda^Q$. The above construction~\eqref{eq:J} therefore includes the canonical one as a particular case.

As remarked above, the control of the limit~\eqref{intro} for such operators doesn't seem within reach of the presently known techniques. However, we shall see in section~\ref{sec:free} below that if $Q$ is the (unbounded) generator of a 1-parameter subgroup of a compact Lie gauge group acting on a finite multiplet of free scalar fields of mass $m \geq 0$, it is possible to provide a different explicit (semi-)local implementation of the flip $W_{O_1,O_2}$ such that the limit~\eqref{intro} can actually be performed for the corresponding generator $J^Q_{O_1,O_2}$ (which is self-adjoint and satisfies~\eqref{eq:localimplemJ} in the same sense as $J_\Lambda^Q$).

The rest of the paper is organized as follows. In section~\ref{sec:bilin} we introduce a new class of test functions spaces and use it to obtain estimates concerning certain free field bilinears; as it is shown in appendix \ref{app:current}, these estimates also allow to establish the existence of the above mentioned unitaries. This is used in 
section \ref{sec:free}, where we go into the study of our models of four dimensional free fields. We focus on the case of a single charged free field with U(1) symmetry, the multiplet case being an easy generalization discussed at the end. We elaborate on the explicit realization of local unitaries implementing the flip automorphisms introduced for the neutral field case in \cite{D'Antoni:1983a}, make use of the multiple commutator theorem in \cite{Frohlich:1977yz} to get an expression for the corresponding local generators of the U(1) symmetry and prove their (essential) self-adjointness on a suitable domain. Finally, convergence of the limit (\ref{intro}) is proved and the constant $c$ there shown to satisfy~\eqref{eq:vol} (in particular it is different from zero).

\section{Test functions spaces and $N$-bounds for free field bilinears}\label{sec:bilin}
We collect here some technical results, needed in the following section, on the extension of bilinear expressions in two commuting complex free scalar fields $\phi_i$, $i=1,2$, and their derivatives, to suitable spaces of tempered distributions. Using this, we will also obtain useful $N$-bounds for such operators.

The Hilbert space $\tsH$ on which the fields $\phi_i$ act is the bosonic second quantization of $K = L^2(\bR^3)\otimes \bC^4$. For $\Phi \in \tsH$, we denote by $\Phi^{(n)}$ its component in $K^{\otimes_S n}$ (the symmetrized $n$-fold tensor power of $K$) and by $\tilde{D}_0$ we indicate the dense space of $\Phi \in \tsH$ such that $\Phi^{(n)} = 0$ for all but finitely many $n\in\bN_0$. Let $\tilde{N}$ be the number operator, defined by $(\tilde{N}\Phi)^{(n)} = n\Phi^{(n)}$ on the domain $D(\tilde{N})$ of vectors $\Phi \in \tsH$ such that $\sum_n n^2\|\Phi^{(n)}\|^2<\infty$. Fixing an orthonormal basis $(e^\tau_i)_{i=1,2}^{\tau = +,-}$ of $\bC^4$, we can identify elements $\Phi \in K^{\otimes_S n}$ with collections $\Phi = (\Phi_{i_1\dots i_n}^{\t_1\dots\t_n})^{\t_1\dots\t_n = +,-}_{i_1\dots i_n = 1,2}$ of functions on $\bR^{3n}$, such that
$\Phi_{i_1\dots i_n}^{\t_1\dots\t_n}(\bp_1,\dots,\bp_n)$ is symmetric for the simultaneous interchange of $(\t_k,i_k,\bp_k)$ and $(\t_h,i_h,\bp_h)$, and
\begin{equation*}
\sum_{\substack{\t_1,\dots,\t_n=+,-\\i_1,\dots,i_n=1,2}}\int_{\bR^{3n}}d\bp_1\dots d\bp_n\,\bigabs{\Phi_{i_1\dots i_n}^{\t_1\dots\t_n}(\bp_1,\dots,\bp_n)}^2 <\infty.
\end{equation*}
We introduce then the operators on $\tsH$
\begin{equation*}
c_i^{\t,-}(\psi) = a(\psi\otimes e^\t_i), \qquad  c_i^{\t,+}(\psi) = a(\bar\psi\otimes e^{-\t}_i)^*,
\end{equation*}
where $\psi \in L^2(\bR^3)$ and $a(\xi)$, $\xi \in K$, is the usual Fock space annihilation operator. Their commutation relations are
\begin{equation*}
[c^{\tau,\sigma}_i(\psi),c^{\rho,\varepsilon}_j(\varphi)] = -\sigma\d_{ij}\d_{\tau,-\r}\d_{\s,-\eps}\int_{\bR^3}d\bp\, \psi(\bp)\varphi(\bp).
\end{equation*}
Introducing also the maps $j_\sigma : \sS(\bR^4) \to L^2(\bR^3)$, $j_\s f(\bp) := \sqrt{2\pi/\omp}\hat{f}(\s\omp,\s\bp)$, $\s=+,-$ (where $\hat{f}(p) = \int_{\bR^4}\frac{dx}{(2\pi)^2}f(x)e^{ipx}$ is the Fourier transform of $f$ and $\omp = \sqrt{|\bp|^2+m^2}$) and the notation $\phi^\dagger_i(f):= \phi_i(\bar{f})^*$, we have
\begin{equation*}
\phi_i(f) = \frac{1}{\sqrt{2}}\sum_{\s=+,-}c^{-,\s}_i(j_\s f), \quad \phi_i^\dagger(f) = \frac{1}{\sqrt{2}}\sum_{\s=+,-}c^{+,\s}_i(j_\s f).
\end{equation*} 
With the notation $\de := \de_0$, we have, for $f \in \sS(\bR^8)$ and $\Phi \in \tilde{D}_0$,
\begin{equation}\label{eq:bilin}
(\wick{\de^l\phi_i\de^k\phi^\dagger_j}(f)\Phi)^{(n)} = \sum_{\s,\eps}\wick{\de^l c^{-,\s}_i\de^k c^{+,\eps}_j}(f)^{(n)} \Phi^{(n-\s-\eps)},
\end{equation}
where $\wick{\de^l c^{-,\s}_i\de^k c^{+,\eps}_j}(f)^{(n)} : K^{\otimes_S (n-\s-\eps)}\to K^{\otimes_S n}$ is a bounded operator whose expression can be obtained from the formal expression of $\phi_i$ in terms of creation and annihilation operators. For instance, if $\Phi \in K^{\otimes_S n}$, 
\begin{gather*}
(\wick{\de^l c^{-,+}_i\de^k c^{+,-}_j}(f)^{(n)} \Phi)^{\tau_1\dots\tau_n}_{i_1\dots i_n}(\bp_1,\dots,\bp_n)= \sum_{r=1}^n \d_{\t_r,+}\d_{i,i_r}i^{l+r} (-1)^k\pi\times\\
\int_{\bR^3}d\bp \,\omp^{k-1/2}\o_m(\bp_r)^{l-1/2}\hat{f}(p_{r,+},-p_+)\Phi^{+\,\t_1\dots\hat{\t}_r\dots\t_n}_{j i_1\dots\hat{i}_r\dots i_n}(\bp,\bp_1,\dots,\hat{\bp}_r,\dots,\bp_n),
\end{gather*}
where the hat over an index means that the index itself must be omitted and where we have introduced the convention (which we will use systematically in the following) of denoting simply by $q_\s \in \bR^4$ the 4-vector $(\s\o_m(\bq),\bq)$, $\s = +,-$.

We now want to show that such operators can be extended to suitable spaces of tempered distributions on $\bR^8$, which in turn are left invariant by the operation induced by the commutator of field bilinears.

\begin{definition}\label{def:Ckl}
We denote by $\hat{\sC}$ the space of functions $f \in C^\infty(\bR^8)$ such that for all $r \in \bN$, $\a,\b \in \bN_0^4$,
\begin{equation*}
\|f\|_{r,\a,\b} = \sup_{(p,q)\in\bR^8}\abs{(1+\abs{\bp+\bq})^r \de_p^\a\de_q^\b f(p,q)}<\infty.
\end{equation*}
Introducing the notation $\tilde{f}(p,q) := f(q,p)$ and the expressions
\begin{equation*}\begin{split}
\big(T^{k,l}(f)\Phi\big)(\bp) &:= \int_{\bR^3}d\bq\, \o_m(\bp)^{k-1/2}\o_m(\bq)^{l-1/2}f(p_+,-q_+)\Phi(\bq),\\
\Phi^{k,l,\s}_f(\bp,\bq) &:= f(\s p_+,\s q_+)\o_m(\bp)^{k-1/2}\o_m(\bq)^{l-1/2},
\end{split}\end{equation*}
where $k,l =0,1$ and $\s = +,-$, we denote by $\hat{\sC}^{k,l}$ the space of functions $f \in \hat{\sC}$ such that $T^{k,l}(\abs{f}), T^{l,k}(\abs{\tilde{f}}) : L^2(\bR^3) \to L^2(\bR^3)$ are bounded operators and $\Phi^{k,l,\s}_f \in L^2(\bR^6)$. Furthermore, we introduce on $\hat{\sC}^{k,l}$ the seminorm
\begin{equation*}
\| f\|_{k,l} := \max\{\|T^{k,l}(\abs{f})\|,\|T^{l,k}(\abs{\tilde{f}})\|,\|\Phi^{k,l,\s}_f\|_{L^2(\bR^6)}\}.
\end{equation*}
\end{definition}

The spaces $\hat{\sC}^{k,l}$ depend also on the mass $m$ appearing in $\o_m$, but we have avoided to indicate this explicitly in order not to burden the notations. It is clear that functions in $\hat{\sC}$ are bounded with all their derivatives and therefore $\hat{\sC}^{k,l} \subset \sS'(\bR^8)$. We denote then by $\sC^{k,l}$ the space of distributions $f\in\sS'(\bR^8)$ such that $\hat{f} \in \hat{\sC}^{k,l}$. It is also easy to verify that $\sS(\bR^8) \subset \hat{\sC}^{k,l}$.

\begin{lemma}\label{lem:boundC}
The expression
\begin{equation}\label{eq:Clk}
\hat{C}^{l,k}(f,g)(p,q):= (-1)^l\pi\sum_{\s=\pm}\s(i\s)^{k+l}\int_{\bR^3}d\bk\, \o_m(\bk)^{l+k-1}f(p,-\s k_+)g(\s k_+,q),
\end{equation}
defines a bilinear map $\hat{C}^{l,k}:\hat{\sC}^{l',l}\times\hat{\sC}^{k,k'}\to\hat{\sC}^{l',k'}$, such that $\|\hat{C}^{l,k}(f,g)\|_{l',k'}\leq 2\pi \|f\|_{l',l}\|g\|_{k,k'}$.
\end{lemma}

\begin{proof}
We start by showing that if $f,g \in \hat{\sC}$ then $\hat{C}^{l,k}(f,g) \in \hat{\sC}$. Setting $\eps = 2/\abs{\bp+\bq}$, and $\boldsymbol{e} = \eps(\bp+\bq)/2$, it is cleary sufficient to show that, as $\eps \to 0$,
\begin{equation}\label{eq:Izeror}
I_{h,r}(\eps) := \int_{\bR^3}\frac{\abs{\bx}^{h/2} d\bx}{(1+\abs{\bx+\eps^{-1}\boldsymbol{e}})^r(1+\abs{\bx-\eps^{-1}\boldsymbol{e}})^r} \leq O(\eps^{s(r,h)}).
\end{equation}
where $h=k+l-1=-1,0,1$, and $s(r,h)\to +\infty$ as $r\to+\infty$. Consider first the case $h=0$. Choosing the $x_3$ axis along $\boldsymbol{e}$ and evaluating the integral in prolate spheroidal coordinates $x_1 = \eps^{-1}\sqrt{(u^2-1)(1-v^2)}\cos\phi$, $x_2 = \eps^{-1}\sqrt{(u^2-1)(1-v^2)}\sin\phi$, $x_3 = \eps^{-1}u v$, one gets
\begin{equation*}
I_{0,r}(\eps) = 2\pi\eps^{2r-3}\bigg[\int_{1+\eps}^{+\infty}du \,J_{r-1}(u) + \eps^2 \int_{1+\eps}^{+\infty}du \,J_r(u)-2\eps\int_{1+\eps}^{+\infty}du \,u J_r(u)\bigg],\
\end{equation*}
where, by recursion,
\begin{equation*}\begin{split}
J_r(u) := \int_{-1}^{1}dv \frac{1}{(u^2-v^2)^r} &= \sum_{k=1}^{r-1}\frac{2(2r-3)\dots(2r-2k+1)}{(2r-2)\dots(2r-2k)}\frac{1}{u^{2k}(u^2-1)^{r-k}} \\
&\quad+ \frac{(2r-3)!!}{(2r-2)!!}\frac{1}{u^{2r-1}}\log\left|\frac{u+1}{u-1}\right|,
\end{split}\end{equation*}
which easily gives estimate~\eqref{eq:Izeror} with $s(r,0) = r-3$. Take now $h=-1$. Dividing the integration region into the subregions $\{\abs{\bx}\leq 1\}$, $\{\abs{\bx}>1\}$ and using the Cauchy-Schwarz inequality in the first integral, one gets
\begin{equation*}
I_{-1,r}(\eps)\leq \left(\int_{\abs{\bx}\leq 1}\abs{\bx}^{-1}d\bx\right)^{1/2}I_{0,2r}(\eps)^{1/2}+I_{0,r}(\eps) \leq O(\eps^{r-3}).
\end{equation*}
Finally, for $h=1$, taking into account the bound $\abs{\bx}^{1/2}/(1+\abs{\bx+\eps^{-1}\boldsymbol{e}})(1+\abs{\bx-\eps^{-1}\boldsymbol{e}}) \leq 1/2$, one gets $I_{1,r}(\eps)\leq O(\eps^{r-4})$.

We now show that if $f\in\hat{\sC}^{l',l}$, $g \in \hat{\sC}^{k,k'}$, then $\hat{C}^{l,k}(f,g) \in \hat{\sC}^{l',k'}$. We introduce the notation $K_\Psi$ to denote the Hilbert-Schmidt operator on $L^2(\bR^3)$ with kernel $\Psi \in L^2(\bR^6)$. It is then easy to verify that, if $\Phi \in L^2(\bR^3)$,
\begin{equation*}\begin{split}
\big\| T^{l',k'}\big(\bigabs{\hat{C}^{l,k}(f,g)}\big)\Phi\big\|_2 &\leq \pi\big(\big\| T^{l',l}(\abs{f})T^{k,k'}(\abs{g})\abs{\Phi}\big\|_2 + \big\|K_{\Phi^{l',l,+}_{\abs{f}}}K_{\Phi^{k,k',-}_{\abs{g}}}\abs{\Phi}\big\|_2\big),\\
\big\| T^{k',l'}\big(\bigabs{\widetilde{\hat{C}^{l,k}(f,g)}}\big)\Phi\big\|_2 &\leq \pi\big(\big\| T^{k',k}(\abs{\tilde{g}})T^{l,l'}(\abs{\tilde{f}})\abs{\Phi}\big\|_2 + \big\|K_{\Phi^{l',l,-}_{\abs{f}}}K_{\Phi^{k,k',+}_{\abs{g}}}\abs{\Phi}\big\|_2\big),
\end{split}\end{equation*}
so that $T^{l',k'}\big(\bigabs{\hat{C}^{l,k}(f,g)}\big)$ and $T^{k',l'}\big(\bigabs{\widetilde{\hat{C}^{l,k}(f,g)}}\big)$ are bounded. Furthermore one has, for $\Psi \in L^2(\bR^6)$,
\begin{equation*}\begin{split}
\big|\big\langle \Phi^{l',k',+}_{\hat{C}^{l,k}(f,g)},\Psi\big\rangle_{L^2(\bR^6)}\big| \leq \pi\Big(&\big\langle\Phi^{k,k',+}_{\abs{g}},\big(T^{l',l}(\abs{f})^*\otimes\Id\big)\abs{\Psi}\big\rangle_{L^2(\bR^6)} \\
+  &\big\langle\Phi^{l',l,+}_{\abs{f}},\big(\Id\otimes T^{k',k}(\abs{\tilde{g}})^*\big)\abs{\Psi}\big\rangle_{L^2(\bR^6)}\Big)\\
\big|\big\langle \Phi^{l',k',-}_{\hat{C}^{l,k}(f,g)},\Psi\big\rangle_{L^2(\bR^6)}\big| \leq \pi\Big(&\big\langle\Phi^{l',l,-}_{\abs{f}},\big(\Id\otimes T^{k,k'}(\abs{g})\big)\abs{\Psi}\big\rangle_{L^2(\bR^6)} \\
+  &\big\langle\Phi^{k,k',-}_{\abs{g}},\big(T^{l,l'}(\abs{\tilde{f}})\otimes\Id\big)\abs{\Psi}\big\rangle_{L^2(\bR^6)}\Big)
\end{split}\end{equation*}
so that by Riesz theorem $\Phi^{l',k',\s}_{\hat{C}^{l,k}(f,g)} \in L^2(\bR^6)$. The bound on $\|\hat{C}^{l,k}(f,g)\|_{l,k}$ now follows at once from the above estimates.
\end{proof}

For $(f,g) \in \sC^{l',l}\times\sC^{k,k'}$ we write $C^{l,k}(f,g) := \hat{C}^{l,k}(\hat{f},\hat{g})\spcheck$.

\begin{proposition}\label{prop:bilin}
The following statements hold for any $i, j \in \{1,2\}, k,l \in \{0,1\}, n \in \bN, \sigma, \eps \in \{+,-\}$, with $n-\s-\eps \geq 0$.
\begin{proplist}{2}
\item The map $f \in \sS(\bR^8) \to \wick{\de^l c^{-,\s}_i\de^k c^{+,\eps}_j}(f)^{(n)} \in B( K^{\otimes_S (n-\s-\eps)},K^{\otimes_S n})$ can be extended to a map (denoted by the same symbol) from $\sC^{l,k}$ to $B( K^{\otimes_S (n-\s-\eps)},K^{\otimes_S n})$, such that
\begin{equation}\label{eq:boundbilc}
\|\wick{\de^l c^{-,\s}_i\de^k c^{+,\eps}_j}(f)^{(n)}\|\leq \pi\|\hat{f}\|_{l,k}(n+2).
\end{equation}
\item For each $f \in \sC^{l,k}$ the operator $\wick{\de^l\phi_i\de^k\phi^\dagger_j}(f)$, defined on $\tilde{D}_0$ by formula~\eqref{eq:bilin}, satisfies
\begin{gather}
\|(\tilde{N}+1)^{-1/2}\wick{\de^l\phi_i\de^k\phi^\dagger_j}(f)(\tilde{N}+1)^{-1/2}\|\leq \upsilon\|\hat{f}\|_{l,k}, \label{eq:Nbound}\\
\|(\tilde{N}+1)^{-1/2}[\tilde{N},\wick{\de^l\phi_i\de^k\phi^\dagger_j}(f)](\tilde{N}+1)^{-1/2}\|\leq \upsilon\|\hat{f}\|_{l,k}, \label{eq:Nboundcomm}
\end{gather}
for some $\upsilon >0$. If furthermore $(f,g) \in \sC^{l',l}\times\sC^{k,k'}$, there holds, on $\tilde{D}_0$,
\begin{equation}\begin{split}\label{eq:comm}
[\wick{\de^{l'}&\phi_{i'}\de^l\phi^\dagger_i}(f), \wick{\de^{k}\phi_{j}\de^{k'}\phi^\dagger_{j'}}(g)] \\
&= \d_{ij}\wick{\de^{l'}\phi_{i'}\de^{k'}\phi^\dagger_{j'}}(C^{l,k}(f,g)) - \d_{i',j'}\wick{\de^{k}\phi_{j}\de^{l}\phi_{i}^\dagger}(C^{k',l'}(g,f))\\
&\quad+i^{l+l'+k+k'}\pi^2\d_{i',j'}\d_{i,j}\Big((-1)^{l+l'}\big\langle\overline{\Phi^{l',l,-}_{\hat f}},\widetilde{\Phi^{k,k',+}_{\hat g}}\big\rangle_{L^2(\bR^6)} \\
&\quad\phantom{i^{l+l'+k+k'}\pi^2\Big[\d_{i',j'}\d_{i,j}\Big(}
- (-1)^{k+k'}\big\langle \overline{\Phi^{l',l,+}_{\hat f}},\widetilde{\Phi^{k,k',-}_{\hat g}}\big\rangle_{L^2(\bR^6)}\Big)\Id. 
\end{split}\end{equation}
\end{proplist}
\end{proposition}

\begin{proof}
(i) Define the contraction operator $\Pi(\psi) : K^{\otimes (n+2)} \to K^{\otimes n}$, $\psi \in K^{\otimes 2}$, by $\Pi(\psi)\psi_1\otimes\dots\otimes\psi_{n+2}=\langle\psi,\psi_1\otimes\psi_2\rangle\psi_3\otimes\dots\otimes\psi_{n+2}$. It is easily seen from the usual expressions of creation and annihilation operators (see, e.g.,~\cite[sec.\ X.7]{Reed:1975a}) that for $f \in \sS(\bR^8)$ 
\begin{equation*}\begin{split}
\wick{\de^l c^{-,+}_i\de^k c^{+,-}_j}(f)^{(n)} &= i^l(-i)^k\pi\sum_{r=1}^n V_r \big((T^{l,k}(\hat{f})\otimes |e^{+}_i\rangle\langle e^+_j|)\otimes\Id\otimes\dots\otimes\Id\big),\\
\wick{\de^l c^{-,+}_i\de^k c^{+,+}_j}(f)^{(n)} &= \frac{i^{l+k}\pi}{\sqrt{n(n-1)}}\sum_{r\neq s}^{1,n} W_{r,s} \Pi\big(\Phi^{l,k,+}_{\hat{f}}\otimes(e^{+}_i\otimes e^{-}_j)\big)^*,\\
\wick{\de^l c^{-,-}_i\de^k c^{+,-}_j}(f)^{(n)} &= (-i)^{l+k}\pi\sqrt{(n+1)(n+2)} \Pi\big(\Phi^{l,k,-}_{\hat{f}}\otimes(e^{-}_i\otimes e^{+}_j)\big),
\end{split}\end{equation*}
where for $\psi_i \in K$, $i=1,\dots,n$,
\begin{equation*}\begin{split}
V_r\psi_1\otimes\dots\otimes\psi_n &= \psi_2\otimes\cdots\underset{\scriptscriptstyle\text{$r$-th place}}{\otimes\,\psi_1\,\otimes}\dots\otimes\psi_n, \\
W_{r,s}\psi_1\otimes\dots\otimes\psi_n &= \psi_3\otimes\cdots\underset{\scriptscriptstyle\text{$r$-th place}}{\otimes\,\psi_1\,\otimes}\cdots\underset{\scriptscriptstyle\text{$s$-th place}}{\otimes\,\psi_2\,\otimes}\cdots\otimes\psi_n.
\end{split}\end{equation*}
Thus the above formulas provide an extension of $\wick{\de^l c^{-,\s}_i\de^k c^{+,\eps}_j}(\cdot)^{(n)}$ to $\sC^{l,k}$ and the bound~\eqref{eq:boundbilc} holds.

(ii) The bounds~\eqref{eq:Nbound}, \eqref{eq:Nboundcomm}, with $\upsilon =4\pi(\sqrt{3}+1)$, follow easily from~\eqref{eq:boundbilc}. Equation~\eqref{eq:comm} is obtained by a straightfoward (if lengthy) calculation, using the above expressions for $\wick{\de^l c^{-,\s}_i\de^k c^{+,\eps}_j}(\cdot)^{(n)}$.
\end{proof}

\begin{Remark}
It is not difficult to see that the above extension of $\wick{\de^l c^{-,\s}_i\de^k c^{+,\eps}_j}(\cdot)^{(n)}$ to $\sC^{l,k}$ is unique in the family of linear maps $S : \sC^{l,k} \to B( K^{\otimes_S (n-\s-\eps)},K^{\otimes_S n})$
which are sequentially continuous when $B( K^{\otimes_S (n-\s-\eps)},K^{\otimes_S n})$ is equipped with the strong operator topology and $\sC^{l,k}$ is equipped with the topology induced by the family of seminorms
\begin{equation*}
\| f \|_{k,l,\Psi} = \max\{ \|T^{k,l}(\abs{\hat f})\Psi\|,\|T^{l,k}(\abs{\tilde{\hat f}})\Psi\|,\|\Phi^{k,l,\s}_{\hat f}\|_{L^2(\bR^6)}\}, \qquad \Psi \in L^2(\bR^3),
\end{equation*}
with respect to which $\sS(\bR^8)$ is sequentially dense in $\sC^{l,k}$. On the other hand we point out the fact that, according to equation~\eqref{eq:comm}, the linear span of extended field bilinears is stable under the operation of taking commutators. Together with proposition~\ref{prop:current} in appendix~\ref{app:current}, this implies that in the construction of the local symmetry generator carried out in the following section, equation~\eqref{eq:Psiseries}, only the above defined extensions are relevant.
\end{Remark}

According to the results in~\cite[sec.\ X.5]{Reed:1975a}, the bounds~\eqref{eq:Nbound}, \eqref{eq:Nboundcomm} imply that $\wick{\de^l\phi_i\de^k\phi^\dagger_j}(f)$ can be extended to an operator, denoted by the same symbol, whose domain contains $D(\tilde{N})$.

\section{Reconstruction of the free field Noether currents}\label{sec:free}
We start by considering the theory of a complex free scalar field $\phi$ of mass $m\geq 0$. The Hilbert space of the theory is the symmetric Fock space $\sH = \Gamma(L^2(\bR^3)\otimes \bC^2)$. As customary, we denote by $D_0 \subset \sH$ the space of finite particle vectors, and by $N$ the number operator $N = d\Gamma(\Id)$, with domain $D(N)$. The local field algebras are defined as usual by
\begin{equation*}
\sF(O) := \big\{ e^{i[\phi(f)+\phi(f)^*]^-}\,:\,f\in\sD(O)\big\}'',
\end{equation*}
and if we consider
\begin{equation*}
V(\theta) := \Gamma\left(\Id\otimes\left(\begin{matrix}e^{i\theta} &0\\0 &e^{-i\theta}\end{matrix}\right)\right),
\end{equation*}
we obtain a continuous unitary representation of U(1) (i.e.\ a $2\pi$-periodic representation of $\bR$) on $\sH$, $\theta \in \bR \to V(\theta)$, which induces a group of gauge automorphisms $\b_\theta := \ad V(\theta)$ of $\sF$ such that $V(\theta)\phi(f)V(\theta) = e^{i\theta}\phi(f)$. We denote by $Q$ the self-adjoint generator of this group. It is easy to see that $\| (N+1)^{-1/2}Q(N+1)^{-1/2}\| \leq 1$ and $[N,Q]=0$, so that thanks to Nelson's commutator theorem (cfr.~\cite[sec.\ X.5]{Reed:1975a}) $D(N) \subset D(Q)$. Furthermore we introduce the unitary operator $Z$ on $\sH$ such that $Z\phi(f)Z^* = -\phi(f)$, $Z\O = \O$.

In order to find an explicit representation of the (semi-)local implementation of the flip automorphism we consider, following~\cite{D'Antoni:1983a}, the doubled theory $O \to \tsF(O) := \sF(O)\vntensor\sF(O)$, generated by the two commuting complex scalar fields $\phi_1(f) := \phi(f)\otimes \Id$, $\phi_2(f):=\Id\otimes\phi(f)$. There is a continuous unitary representation of U(1)  on $\tsH = \sH\otimes\sH$, $\zeta \in \bR \to  Y(\zeta)$, which induces a group of gauge automorphisms $\g_\zeta := \ad Y(\zeta)$ of $\tsF$ such that
\begin{equation}\label{eq:doublesymmetry}\begin{split}
\g_\zeta(\phi_1(f)) &= \cos \zeta \,\phi_1(f) - \sin \zeta\,\phi_2(f),\\
\g_\zeta(\phi_2(f)) &= \sin\zeta\,\phi_1(f) + \cos\zeta\, \phi_2(f).
\end{split}\end{equation}
In proposition~\ref{prop:current} in appendix~\ref{app:current} it is shown that the Noether current of this U(1) symmetry
\begin{equation}\label{eq:currentdouble}
J_\m(x) = \phi_1(x)\de_\m\phi_2(x)^* + \phi_1(x)^*\de_\m\phi_2(x) - \de_\m\phi_1(x)\phi_2(x)^* - \de_\m\phi_1(x)^*\phi_2(x)
\end{equation}
is a well-defined Wightman field that when smeared with an $h \in \sS_\bR(\bR^4)$ gives an operator which is essentially self-adjoint on $D(\tilde{N})$, and generates a group of unitaries which locally implements the symmetry: given 3-dimensional open balls $B_r, B_{r+\d}$ of radiuses $r, r+\d > 0$ centered at the origin together with functions $\varphi \in \sD_\bR(B_{r+\d-\tau})$, $\psi \in \sD_\bR((-\tau,\tau))$ such that $\tau<\delta/2$, $\varphi(\bx) = 1$ for each $\bx \in B_{r+\tau}$ and $\int_\bR \psi = 1$, it holds that
\begin{equation}\label{eq:localimplemdouble}
e^{i\zeta J_0(\psi\otimes \varphi)}\in \tsF(O_{r+\d}),\quad e^{i\zeta J_0(\psi\otimes \varphi)}F e^{-i\zeta J_0(\psi\otimes\varphi)} = \g_\zeta(F),\quad \forall \,F \in \tsF(O_r),
\end{equation}
where $O_r$, $O_{r+\d}$ are the double cones with bases $B_r$, $B_{r+\d}$ respectively. It follows then easily that setting $h_\l := \psi_\l\otimes\varphi_\l$ with  $\varphi_\l(\bx) = \varphi(\l^{-1}\bx)$ and $\psi_\l(t) = \l^{-1}\psi(\l^{-1} t)$, the unitary operator 
\begin{equation}\label{eq:W}
\WlOrd := (\Id\otimes Z)e^{i\frac{\pi}{2}J_0(h_\l)}\in\sF(\l O_{r+\d})\vntensor B(\sH),
\end{equation}
is a semi-local implementation of the flip automorphism on $\tsF(\l O_r)$ for each $\l > 0$.  In what follows, we will keep the functions $\varphi$, $\psi$ fixed and we will assume that $\varphi(R \bx) = \varphi(\bx)$ for each $R \in O(3)$.

For a function $h \in \sS(\bR^4)$, we introduce the distribution $h_\d \in \sS'(\bR^8)$ defined by $h_\d(x,y) = h(x)\d(x-y)$ (i.e.\ $\langle h_\d, f\rangle = \int_{\bR^4}dx\,h(x)f(x,x)$ for $f \in \sS(\bR^8)$).

\begin{proposition}\label{prop:psiseries}
Let the operator $\WlOrd$ be defined as above. The operator $\XilO(Q)$ defined on $D(N)$ by
\begin{equation}\label{eq:Psil}
\XilO(Q)\Phi = P_1 \WlOrd (\Id\otimes Q)\WlOrd^*\Phi\otimes \O, \qquad \Phi \in D(N), 
\end{equation}
where $P_1 \Phi_1\otimes\Phi_2 = \langle \O,\Phi_2\rangle \Phi_1$, is essentially self-adjoint. Furthermore, there are distributions $K^{l,k}_{n,m}(\l) \in \sC^{l,k}$, $n\in\bN$, $l,k = 0,1$, $m \geq 0$, defined recursively by
\begin{gather}
K^{1,0}_{1,m}(\l) = -K^{0,1}_{1,m}(\l) := (h_\l)_\d, \qquad K^{0,0}_{1,m}(\l) = K^{1,1}_{1,m}(\l) = 0,\label{eq:Krecursinit}\\
\begin{split}
K^{l,k}_{n+1,m}(\l) &= i(-1)^n\sum_{r=0}^{1} \Big[(-1)^{l+1} C^{1-l,r}\big((h_\l)_\d,K^{r,k}_{n,m}(\l)\big)\\
&\phantom{i(-1)^n\sum_{r=0}^{1} \Big[}+ (-1)^k C^{r,1-k}\big(K^{l,r}_{n,m}(\l),(h_\l)_\d\big)\Big],\label{eq:Krecurs}
\end{split}\end{gather}
such that, for all $\Phi\in D(N)$,
\begin{equation}\label{eq:Psiseries}
\XilO(Q)\Phi = \sum_{n=1}^{+\infty}\frac{\pi^{2n}}{4^n(2n)!}\bigg[\sum_{l,k}^{0,1}\wick{\de^l\phi\de^k\phi^\dagger}\big(K^{l,k}_{2n,m}(\l)\big)\Phi 
\bigg],
\end{equation}
the series being absolutely convergent for all $\l\in(0,1]$.
\end{proposition}

\begin{proof} We start by observing that, for all $\Phi \in \sH$ for which the right hand side of~\eqref{eq:Psil} is defined, one has
\begin{equation}\label{eq:Psilalt}
\XilO(Q)\Phi=P_1 e^{i\frac{\pi}{2}J_0(h_\l)}(\Id\otimes Q)e^{-i\frac{\pi}{2}J_0(h_\l)}\Phi\otimes\O.
\end{equation}
It follows from this formula that $\XilO(Q)$ is well-defined (and symmetric) on $D(N)$: according to formula~\eqref{eq:currentf} in appendix~\ref{app:current} for $J_0(h_\l)$, proposition~\ref{prop:bilin}(ii) and~\cite[lemma 2]{Frohlich:1977yz}, we have $e^{i\frac{\pi}{2}J_0(h_\l)}D(\tilde{N}) \subset D(\tilde{N})$ and $D(N) \subset D(Q)$ as remarked above.

Recalling now the definition of $Q$ one has on $D(\tilde{N})$
\begin{equation*}\begin{split}
Q_1(\l) &:= i[J_0(h_\l),\Id\otimes Q] = \sum_{j=1}^2[\wick{\de\phi_j\phi^\dagger_{j'}}((h_\l)_\d) - \wick{\phi_j\de\phi^\dagger_{j'}}((h_\l)_\d)] \\
&= \sum_{j=1}^2 \sum_{l,k}^{0,1}\wick{\de^l\phi_j\de^k\phi^\dagger_{j'}}\big(K^{l,k}_{1,m}(\l)\big),
\end{split}\end{equation*}
where $j' = 3-j$. Proceeding now inductively using formula~\eqref{eq:comm}, one verifies that there are operators $Q_n(\l)$ such that, on $\tilde{D}_0$,
\begin{gather}
Q_{n+1}(\l) = i[J_0(h_\l),Q_n(\l)],\label{eq:Qninductive}\\
\begin{split}
Q_{2n}(\l) &= \sum_{j=1}^2 \sum_{l,k}^{0,1}(-1)^{j+1}\wick{\de^l\phi_j\de^k\phi^\dagger_{j}}\big(K^{l,k}_{2n,m}(\l)\big) 
,\\
Q_{2n+1}(\l) &= \sum_{j=1}^2 \sum_{l,k}^{0,1}\wick{\de^l\phi_j\de^k\phi^\dagger_{j'}}\big(K^{l,k}_{2n+1,m}(\l)\big),
\end{split}\end{gather}
where the distributions $K^{l,k}_{n,m}(\l) \in \sC^{l,k}$ satisfy~\eqref{eq:Krecurs}.
It is also easy to verify inductively that the distributions $K^{l,k}_{n,m}(\l)$ are real ($g \in \sS'$ being real if $\langle g, f\rangle = \overline{\langle g, \bar{f}\rangle}$), so that $Q_n(\l)$ is symmetric. Arguing again by induction, it follows from~\eqref{eq:Krecurs} and lemma~\ref{lem:boundC}, that 
\begin{equation*}
\| \hat{K}^{l,k}_{n,m}(\l)\|_{l,k} \leq (8\pi)^{n-1} \left(\max\big\{\|\widehat{(h_\l)_\d}\|_{0,1}, \|\widehat{(h_\l)_\d}\|_{1,0}\big\}\right)^n\leq (8\pi)^{n-1} \|h\|_{\sS}^n,
\end{equation*}
where $\|h\|_{\sS}$ is some fixed Schwartz norm of $h$. The last inequality above follows from lemma~\ref{lem:bound} and from the observation that, switching for a moment to the notation $\|\cdot\|_{l,k}^{(m)}$ in order to make explicit the dependence on the mass $m$ of the seminorms $\|\cdot\|_{l,k}$, one has
\begin{equation*}
\| \widehat{(h_\l)_\d} \|^{(m)}_{l,1-l} = \| \hat{h}_\d \|^{(\l m)}_{l,1-l}, \qquad l=0,1.
\end{equation*}
Using now the bounds in proposition~\ref{prop:bilin} and the results in~\cite[sec.\ X.5]{Reed:1975a}, we see that $Q_n(\l)$ can be extended to an operator (denoted by the same symbol) which is essentially self-adjoint on any core for $\tilde{N}$. The domain $\tilde{D}_0$ being such a core, equation~\eqref{eq:Qninductive} can be assumed to hold weakly on $D(\tilde{N})\times D(\tilde{N})$ and we are therefore in the position of applying~\cite[thm.\ 1$_\infty$]{Frohlich:1977yz} to obtain
\begin{equation*}
e^{i\frac{\pi}{2}J_0(h_\l)}(\Id\otimes Q)e^{-i\frac{\pi}{2}J_0(h_\l)} = \Id\otimes Q +\sum_{n=1}^{+\infty}\frac{1}{n!}\left(\frac{\pi}{2}\right)^nQ_n(\l)
\end{equation*}
and the series converges strongly absolutely on $D(\tilde{N})$. Combining this with~\eqref{eq:Psilalt}, and the fact that $P_1 \wick{\de^l\phi_j\de^k\phi^\dagger_{j'}}\big(K^{l,k}_{2n+1,m}(\l)\big)\Phi\otimes \O = 0 = P_1 \wick{\de^l\phi_2\de^k\phi^\dagger_{2}}\big(K^{l,k}_{2n,m}(\l)\big)\Phi\otimes \O$, equation~\eqref{eq:Psiseries} readily follows, upon identification of $\phi_1(f) = \phi(f)\otimes\Id$ with $\phi(f)$.

It remains to prove that $\XilO(Q)$ is essentially self-adjoint on $D(N)$, but this again follows from the easily obtained $N$-bounds
\begin{equation}\label{eq:NboundsXi}\begin{split}
\| (N+1)^{-1/2}\XilO(Q)(N+1)^{-1/2}\| &\leq \gamma \cosh\left(4\pi^2\|h\|_{\sS}\right), \\
\| (N+1)^{-1/2}\big[N,\XilO(Q)\big](N+1)^{-1/2}\| &\leq \gamma \cosh\left(4\pi^2\|h\|_{\sS}\right),
\end{split}\end{equation}
where $\gamma > 0$ is a suitable numerical constant. 
\end{proof}

We now show that the unitary group generated by the operator $\XiO(Q)$ defined in the above proposition provides a local implementation of the U(1) symmetry.

\begin{proposition}
For each $\theta \in \bR$ and $F \in \sF(O_r)$ there holds:
\begin{equation*}
e^{i \theta \XiO(Q)} \in \sF(O_{r+\d}),\qquad e^{i\theta\XiO(Q)}Fe^{-i\theta\XiO(Q)} = \b_\theta(F).
\end{equation*}
\end{proposition}

\begin{proof}
Since the free field enjoys Haag duality property, it is sufficient to show that 
\begin{equation*}
e^{i \theta \XiO(Q)} e^{i[\phi(f)+\phi(f)^*]^-} e^{-i \theta \XiO(Q)} = e^{i[\phi(f)+\phi(f)^*]^-}
\end{equation*}
if $\supp f \subset O_{r+\d}'$ and that
\begin{equation*}
e^{i \theta \XiO(Q)} e^{i[\phi(f)+\phi(f)^*]^-} e^{-i \theta \XiO(Q)} = e^{i[e^{i\theta}\phi(f)+e^{-i\theta}\phi(f)^*]^-} 
\end{equation*}
if $\supp f \subset O_r$. Applying once again~\cite[thm.\ 1$_\infty$]{Frohlich:1977yz} and keeping in mind the previously obtained $N$-bounds for $\XiO(Q)$, eq.~\eqref{eq:NboundsXi},  one sees that in order to achieve this, it is enough to show that for all $\Phi_1, \Phi_2 \in D(N)$
\begin{equation}\label{eq:weakcommext}
\langle\XiO(Q)\Phi_1, \phi(f) \Phi_2\rangle - \langle\phi(f)^*\Phi_1, \XiO(Q)\Phi_2\rangle = 0
\end{equation}
for $\supp f \subset O_{r+\d}'$ and 
\begin{equation}\label{eq:weakcommint}
\langle\XiO(Q)\Phi_1, \phi(f) \Phi_2\rangle - \langle\phi(f)^*\Phi_1, \XiO(Q)\Phi_2\rangle = \langle \Phi_1, \phi(f)\Phi_2\rangle
\end{equation}
for $\supp f \subset O_r$. In order to prove the latter equation we compute
\begin{equation*}\begin{split}
\langle \XiO&(Q)\Phi_1,\phi(f)\Phi_2 \rangle \\
&= \langle (\Id\otimes Q)e^{-i\frac{\pi}{2}J_0(h)}\Phi_1\otimes\O, e^{-i\frac{\pi}{2}J_0(h)}(\phi(f)\otimes\Id)\Phi_2\otimes\Omega\rangle\\
&= \langle (\Id\otimes Q)e^{-i\frac{\pi}{2}J_0(h)}\Phi_1\otimes\O, (\Id\otimes\phi(f))e^{-i\frac{\pi}{2}J_0(h)}\Phi_2\otimes\Omega\rangle\\
&= \langle (\Id\otimes \phi(f)^*)e^{-i\frac{\pi}{2}J_0(h)}\Phi_1\otimes\O, (\Id\otimes Q)e^{-i\frac{\pi}{2}J_0(h)}\Phi_2\otimes\Omega\rangle\\
&\qquad +\langle e^{-i\frac{\pi}{2}J_0(h)}\Phi_1\otimes\O,(\Id\otimes\phi(f))e^{-i\frac{\pi}{2}J_0(h)}\Phi_2\otimes\O\rangle \\
&= \langle \phi(f)^*\Phi_1,\XiO(Q)\Phi_2\rangle + \langle\Phi_1,\phi(f)\Phi_2\rangle,
\end{split}\end{equation*}
where in the second and fourth equalities we used~\eqref{eq:doublesymmetry} and~\eqref{eq:localimplemdouble}, and in the third equality
the fact that, as noted in the proof of proposition~\ref{prop:psiseries}, $e^{-i\frac{\pi}{2}J_0(h)}\Phi_i\otimes\O \in D(\tilde N)$ and that for $\tilde \Phi_1, \tilde \Phi_2 \in D(\tilde N)$ there holds
\begin{equation*}
\langle (\Id\otimes Q)\tilde\Phi_1, (\Id\otimes\phi(f))\tilde\Phi_2\rangle - \langle (\Id\otimes \phi(f)^*)\tilde\Phi_1, (\Id\otimes Q)\tilde\Phi_2\rangle = \langle\tilde\Phi_1,(\Id\otimes\phi(f))\tilde\Phi_2\rangle
\end{equation*}
which in turns is an easy consequence of the commutation relation
\begin{equation*}
[Q,\phi(f)]\Phi = \phi(f)\Phi, \qquad \Phi \in D(N),
\end{equation*}
of the fact that $\tilde N$ is the closure of $N\otimes\Id + \Id \otimes N$ and of the $\tilde N$-bounds holding for $\Id \otimes Q$ and $\Id \otimes \phi(f)$. The proof of~\eqref{eq:weakcommext} being analogous, we get the statement.
\end{proof}

In the following lemma, we collect some properties of the distributions $K^{l,k}_{n,m} := K^{l,k}_{n,m}(1)$ which will be needed further on. We will use systematically the notations
\begin{gather*}
\| f\|_\a := \sup_{p \in \bR^4} (1+|p_0|+|\bp|)^\a |f(p)|,\qquad \|\varphi\|_\a := \sup_{\bp \in \bR^3}(1+|\bp|)^\a |\varphi(\bp)|,\\
\|\psi\|_{1,\infty} := \max\{\|\psi\|_\infty,\|\psi'\|_\infty\}, \qquad \|\varphi\|_{1,\a}:= \max\{\|\varphi\|_\a, \|\de_1\varphi\|_\a,\dots,\|\de_3\varphi\|_\a\},
\end{gather*}
for $f \in \sS(\bR^4)$, $\varphi \in \sS(\bR^3)$, $\psi \in \sS(\bR)$ and $\a > 0$. 

\begin{lemma}\label{lem:Kprop}
The following statements hold.
\begin{proplist}{3}
\item \label{it:symm}The functions $\hat K^{l,k}_{n,m}$ enjoy the following symmetry properties:
\begin{equation}\label{eq:Ksymm}
\hat K^{l,k}_{n,m}(p,q) = - \hat K^{k,l}_{n,m}(q,p), \qquad \hat K^{l,k}_{n,m}(p_0, R\bp,q_0, R\bq) = \hat K^{l,k}_{n,m}(p,q)
\end{equation}
for all $p = (p_0,\bp), q = (q_0,\bq) \in \bR^4$, and all $R \in O(3)$.
\item \label{it:bound}Given $\a > 5$ there exists a constant $C_1 > 0$ such that, uniformly for all $m \in [0,1]$ and all smearing functions  $\varphi \in \sD_\bR(B_{r+\d-\tau})$, $\psi \in \sD_\bR((-\tau,\tau))$,
\begin{equation}\label{eq:Kbound}
\left| \hat K^{l,k}_{n,m}(p,q) \right| \leq \frac{C_1^{n-1}}{4\pi^2}\|\hat \psi\|_\infty^n\|\hat \varphi\|_\a^n(1+|\bp|)^{2-l}(1+|\bq|)^{2-k}, \qquad n \in \bN,
\end{equation}
for all $p = (p_0,\bp), q = (q_0,\bq) \in \bR^4$.
\item \label{it:cont}For each $n \in \bN$, the function $(p,q,m) \in \bR^8\times [0,1] \to \hat K^{l,k}_{n,m}(p,q)$ is continuous.
\item \label{it:deriv} For each $n \in \bN$, the function $(p,q,m) \in \bR^8\times[0,1/e] \to  \hat K^{l,k}_{n,m}(p,q)$ is of class $C^1$. Moreover, given $\a > 5$ there exists a constant $C_2 \geq C_1$ such that uniformly for all $m \in [0,1/e]$ and all smearing functions  $\varphi \in \sD_\bR(B_{r+\d-\tau})$, $\psi \in \sD_\bR((-\tau,\tau))$,
\begin{align}
\left| \frac{\de}{\de u_\mu}\hat K^{l,k}_{n,m}(p,q) \right| &\leq \frac{C_1^{n-1}}{4\pi^2}\|\hat \psi\|_{1,\infty}^n\|\hat \varphi\|_{1,\a}^n(1+|\bp|)^{2-l}(1+|\bq|)^{2-k}, \label{eq:Kboundderk}\\
\left| \frac{\de}{\de m}\hat K^{l,k}_{n,m}(p,q) \right| &\leq m|\log m|\frac{C_2^{n-1}}{4\pi^2}\|\hat \psi\|_{1,\infty}^n\|\hat \varphi\|_{1,\a}^n(1+|\bp|)^{2-l}(1+|\bq|)^{2-k}, \label{eq:Kboundderm}
\end{align}
for all $p = (p_0,\bp), q = (q_0,\bq) \in \bR^4$, and where $u$ in~\eqref{eq:Kboundderk} is $p$ or $q$.
\end{proplist}
\end{lemma}

\begin{proof}
\ref{it:symm} Both properties in~\eqref{eq:Ksymm} follow easily by induction from the recursive definition of $\hat K^{l,k}_{n,m}$, taking into account rotational invariance of the function $\varphi$.

\ref{it:bound} We start by observing that, by interchanging $\bk$ with $-\bk$ in the $\s = -1$ summand, formula~\eqref{eq:Clk} can be rewritten as
\begin{equation}\label{eq:Clkalt}
\hat{C}^{l,k}(f,g)(p,q):= (-1)^l\pi\sum_{\s=\pm}\s(i\s)^{k+l}\int_{\bR^3}d\bk\, \o_m(\bk)^{l+k-1}f(p,- k_\s)g(k_\s,q),
\end{equation}
where we recall that $k_\s = (\s \o_m(\bk),\bk)$. Since $\a > 5$, there exists a fixed constant
\begin{equation*}
B_1 > \int_{\bR^3}\frac{d\bk}{|\bk|(1+|\bp-\bk|)^\a},  \int_{\bR^3}\frac{|\bk|^sd\bk}{(1+|\bk|)^\a}, \quad s=0,1,2, \; \bp \in \bR^3.
\end{equation*}
It is then easily computed that for $h=-1,0,1$, $j=1,2$ and $m \in [0,1]$,
\begin{equation*}
\int_{\bR^3}d\bk\frac{\o_m(\bk)^{h}(1+|\bk|)^j}{(1+|\bp-\bk|)^\a}\leq 7B_1 (1+|\bp|)^{h+j},
\end{equation*}
so that estimate~\eqref{eq:Kbound} follows by induction from~\eqref{eq:Krecurs} and the above expression for $\hat{C}^{l,k}$, where one should define
 $C_1 := 14 B_1/\pi$ and keep in mind that $\hat h_\delta (p,q) = \frac{1}{4\pi^2}\hat \psi(p_0+q_0)\hat\varphi(\bp+\bq)$.

\ref{it:cont} Using~\eqref{eq:Kbound} and the fact that $\hat h \in \sS(\bR^4)$, we obtain a bound to the integrands in $\hat C^{1-l,r}(\hat h_\d,\hat K^{r,k}_{n,m})$ and $\hat C^{r,1-k}(\hat K^{l,r}_{n,m},\hat h_\d)$ with an integrable function of $\bk$, uniformly for $(p,q,m)$ in a prescribed neighbourhood of any given $(\bar p,\bar q,\bar m) \in \bR^8\times[0,1]$. By a straightforward application of Lebesgue's dominated convergence theorem, the continuity of $(p,q,m) \to \hat K^{l,k}_{n,m}(p,q)$ follows then by induction from the recursive relation~\eqref{eq:Krecurs}.

\ref{it:deriv} Since $\hat K^{l,k}_{n,m} \in \hat\sC^{l,k}$, we already know that it is differentiable with respect to the components of $p$ and $q$. The estimate~\eqref{eq:Kboundderk} and the continuity of 
$(p,q,m) \to \frac{\de}{\de u_\mu}\hat K^{l,k}_{n,m}(p,q)$ then follow by an easy adaptation of the inductive arguments of~\ref{it:bound} and~\ref{it:cont}, using also~\eqref{eq:Kbound}. In order to show that $\hat K^{l,k}_{n,m}$ is continuously differentiable in $m$ and satisfies~\eqref{eq:Kboundderm}, we proceed again by induction using~\eqref{eq:Krecurs}. The $m$-derivative of the integrands in $\hat C^{1-l,r}(\hat h_\d,\hat K^{r,k}_{n,m})$ is given, apart from numerical constants, by
\begin{multline*}
\frac{m(r-l)}{\o_m(\bk)^{2+l-r}}\hat h(p-k_\s)\hat K^{r,k}_{n,m}(k_\s,q)-\frac{\s m}{\o_m(\bk)^{1+l-r}}\left[\de_0 \hat h(p-k_\s)\hat K^{r,k}_{n,m}(k_\s,q) \right. \\
- \left. \hat h(p-k_\s) \frac{\de\hat K^{r,k}_{n,m}}{\de p_0}(k_\s,q)\right] 
+ \o_m(\bk)^{r-l}\hat h(p- k_\s) \frac{\de}{\de m}\hat K^{r,k}_{n,m}(k_\s,q).
\end{multline*}
It is now straightforward to verify, using~\eqref{eq:Kbound}, \eqref{eq:Kboundderk} and the inductive hypotesis~\eqref{eq:Kboundderm}, that it is possible to bound the last three terms in the above expression with an integrable function of $\bk$, uniformly for $(p,q,m)$ in a given neighbourhood of a fixed $(\bar p, \bar q, \bar m) \in \bR^8 \times [0,1/e]$. The same reasoning also applies to the first term when $2+l-r < 3$ and also when $2+l-r= 3$ for $|\bk| \geq 1/2$. For $|\bk| \leq 1/2$ and $2+l-r=3$ the first term can be bounded uniformly in a neighbourhood of $(\bar p, \bar q)$ by the function $m(m+|\bk|)^{-3}$, apart from a constant (depending on the chosen neighbourhood). By maximizing the function $x \mapsto x^3 |\log x|^{\b}/(m+x)^3$ in the interval $[0,1/2]$, with $\b > 1$, one finds the bound
\begin{equation*}
\frac{m}{(m+|\bk|)^3} \leq \left(\frac{3}{\b}\right)^3\frac{m W_0\big(\frac{\b}{3m}e^{-\beta/3}\big)^3}{|\bk|^3 \big|\log |\bk|\big|^\b},
\end{equation*}
where $W_0$ is the principal branch of Lambert's $W$ function~\cite{Corless:1996}. From the asymptotic expansion of $W_0$ given in~\cite[eq.\ (4.20)]{Corless:1996} it is then easily seen that the numerator on the right hand side converges to 0 as $m \to 0$; since the function $\bk \to |\bk|^{-3} \smash{\big|\log |\bk|\big|}^{-\b}$ is integrable for $|\bk|\leq 1/2$, interchangeability of derivation with respect to $m$ and integration with respect to $\bk$ in $\hat C^{1-l,r}(\hat h_\d,\hat K^{r,k}_{n,m})$ for all values of $l,r,k = 0,1$ follows. A completely analogous argument applies of course to $\hat C^{r,1-k}(\hat K^{l,r}_{n,m},\hat h_\d)$, so that we conclude that $\hat K^{l,k}_{n+1,m}$ is continuously differentiable in $m$. To complete the inductive step, it remains to be shown that estimate~\eqref{eq:Kboundderm} holds for $\frac{\de}{\de m}\hat K^{l,k}_{n+1,m}$. In order to do that, we argue again in a similar way as in~\ref{it:bound} by choosing constants $B_2, B_3 > 0$ such that
\begin{alignat*}{2}
B_2 &\geq \int_{\bR^3}\frac{d\bk}{|\bk|^t(1+|\bp-\bk|)^\a},  \int_{\bR^3}\frac{|\bk|^sd\bk}{(1+|\bk|)^\a}, &\quad &s=0,1,\,t=0,1,2,\;\bp\in\bR^3\\
B_3 &\geq \log\big(1+\sqrt{1+m^2}\big)-\frac{1}{\sqrt{1+m^2}}, & &m\in[0,1/e].
\end{alignat*}
Taking now into account the identity
\begin{equation*}
\int_0^1 \frac{x^2\,dx}{(m^2+x^2)^{3/2}}= \log\big(1+\sqrt{1+m^2}\big)-\frac{1}{\sqrt{1+m^2}}-\log m,
\end{equation*}
it is easy to verify that the estimate
\begin{multline*}
\left| \frac{\de}{\de m}\hat C^{1-l,r}(\hat h_\d,\hat K^{r,k}_{n,m})(p,q)\right| \leq\\
\frac{m|\log m|}{8\pi^3}\big[16\pi(1+B_3)+16B_2+7B_1\big]C_2^{n-1}\|\hat\psi\|_{1,\infty}^{n+1}\|\hat \varphi\|_{1,\a}^{n+1}(1+|\bp|)^{2-l}(1+|\bq|)^{2-k},
\end{multline*} 
holds for all $m \in [0,1/e]$ together with a similar one for $\frac{\de}{\de m}\hat C^{r,1-k}(\hat K^{l,r}_{n,m},\hat h_\d)$. Choosing $C_2 := \frac{2}{\pi}\big[16\pi(1+B_3)+16B_2+7B_1\big] \geq C_1$, one finally gets~\eqref{eq:Kboundderm} for $\hat K^{l,k}_{n+1,m}$.
\end{proof}

In the next theorem, which is our main result, we denote by $D_{0,\sS}$ the dense subspace of $\sH$ of finite particle vectors such that the $n$-particle wave functions are in $\sS(\bR^{3n})$ for each $n \in \bN$.  

\begin{theorem}
There holds, for each $f \in \sS(\bR^4)$ and each $\Phi \in D_{0,\sS}$,
\begin{equation}\label{eq:limpsi}
\lim_{\l \to 0}\frac{1}{\l^3}\int_{\bR^4}dx\,f(x)
\a_x(\XilO(Q))\Phi
= c j_0(f)\Phi,
\end{equation}
where $j_0(f) = \wick{\de\phi\phi^\dagger-\phi\de\phi^\dagger}(f_\d)$ is the Noether current associated to the U(1) symmetry of the charged Klein-Gordon field of mass $m \geq 0$ smeared with the test function $f$ and
\begin{equation}\label{eq:cseries}
c = -(2\pi)^4\sum_{n=1}^{+\infty} \frac{\pi^{2n}}{4^n(2n)!}\left[\hat K^{0,1}_{2n,0}(0,0)+i\frac{\de\hat K^{0,0}_{2n,0}}{\de p_0}(0,0)\right].
\end{equation}
\end{theorem}

\begin{proof}
Since $D_{0,\sS}$ is translation invariant and contained in $D(N)$, according to proposition \ref{prop:bilin} and the estimates given in the proof of proposition~\ref{prop:psiseries} there exists a $\upsilon> 0$ such that, for each $x \in \bR^4$,
\begin{equation*}
\| \a_x\big(\wick{\de^l\phi\de^k\phi^\dagger}\big(K^{l,k}_{2n,m}(\l)\big)\big)\Phi\| \leq \upsilon (8\pi)^{2n-1}\|h\|_\sS^{2n}\|(N+1)\Phi\|,
\end{equation*}
and
\begin{eqnarray*}
\qquad\| \a_x\big(\wick{\de^l\phi\de^k\phi^\dagger}\big(K^{l,k}_{2n,m}(\l)\big)\big)\Phi- \a_y\big(\wick{\de^l\phi\de^k\phi^\dagger}\big(K^{l,k}_{2n,m}(\l)\big)\big)\Phi \| \leq \\
\leq\upsilon (8\pi)^{2n-1}\|h\|_{\sS}^{2n}\| (U(x)^*-U(y)^*)(N+1)\Phi\| \qquad\qquad\\
+\|(U(x)-U(y))\wick{\de^l\phi\de^k\phi^\dagger}\big(K^{l,k}_{2n,m}(\l)\big)U(y)^*\Phi\| ,\qquad\quad\;
\end{eqnarray*}
so that the function $x \to \a_x(\XilO(Q))\Phi$ is continuous and bounded in norm for each $\Phi \in D_{0,\sS}$, the integral in~\eqref{eq:limpsi} exists in the Bochner sense and furthermore it is possible to interchange the integral and the series.

Given now $K \in \sC^{l,k}$, it is easy to see that the pointwise product $\hat K \hat f_\delta$ still belongs to $\hat\sC^{l,k}$ and $\| \hat K \hat f_\delta \|_{l,k} \leq  \frac{1}{(2\pi)^2}\| \hat f\|_\infty \| \hat K\|_{l,k}$ so that we can define $K * f := (2\pi)^4(\hat K \hat f_\d)\spcheck\in\sC^{l,k}$. It is then straightforward to check that
\begin{equation*}
\int_{\bR^4}dx\,f(x)\a_x\big(\wick{\de^l\phi\de^k\phi^\dagger}\big(K^{l,k}_{2n,m}(\l)\big)\big)\Phi = \,\,\wick{\de^l\phi\de^k\phi^\dagger}\big(K^{l,k}_{2n,m}(\l)*f\big)\Phi.
\end{equation*}
Furthermore one has $\hat K^{l,k}_{2n,m}(\l)(p,q) = \l^{2+l+k}\hat K^{l,k}_{2n,\l m}(\l p, \l q)$ and, with the notation $(\d_\l K)\sphat(p,q) = \hat K(\l p, \l q)$, we see that we are left with the calculation of
\begin{equation}\label{eq:limitseries}
 \lim_{\l \to 0}\sum_{l,k}^{0,1}\lambda^{l+k-1}\sum_{n=1}^{+\infty}\frac{\pi^{2n}}{4^n(2n)!}\wick{\de^l\phi\de^k\phi^\dagger}\big(\d_\l K^{l,k}_{2n,\l m}*f\big)\Phi.
 \end{equation}
 
 As a first step in this calculation, we show that it is possible to interchange the limit and the series. Of course, it is sufficient to consider vectors $\Phi$ with vanishing $n$-particles components except for $n=N$ with any fixed $N \in \bN$. For simplicity, we will give here only the relevant estimates in the case $m>0$, the case $m=0$ being treated in a similar way. Using then the notations for creation and annihilation operators and for wave functions introduced in section~\ref{sec:bilin} and the formulas in the proof of proposition~\ref{prop:bilin}, we have
\begin{multline*}
\big\| \wick{\de^lc^{-,+}\de^k c^{+,-}}(\d_\l K^{l,k}_{2n,\l m}*f)^{(N)}\Phi\big\| \leq \\
16\pi^5N\big\|\big( \big(T^{l,k}\big((\d_\l K^{l,k}_{2n,\l m})\sphat \hat f_\d\big)\otimes |e^+\rangle	\langle e^+|\big)\otimes\Id\otimes\dots\otimes\Id\big)\Phi\big\|,
\end{multline*}
together with the estimate, for $\l \in [0,1/m]$,
\begin{multline*}
\big|\big[ \big(\big(T^{l,k}\big((\d_\l K^{l,k}_{2n,\l m})\sphat \hat f_\d\big)\otimes |e^+\rangle	\langle e^+|\big)\otimes\Id\otimes\dots\otimes\Id\big)\Phi\big]^{\tau_1\dots\tau_N}(\bp_1,\dots,\bp_N)\big| \leq \\
\frac{C_1^{n-1}B_1}{4\pi^2}\|\hat \psi\|_\infty^n\|\hat \varphi\|_\a^n\|\hat f\|_\b
\frac{(1+|\bp_1|)^{2-l}\o_m(\bp_1)^{l-1/2}}{(1+|\bp_2|)^\a\dots(1+|\bp_N|)^\a}
\int_{\bR^3}d\bq \frac{\o_m(\bq)^{k-1/2}(1+|\bq|)^{2-k}}{(1+|\bq|)^\g(1+|\bp_1-\bq|)^\b},
\end{multline*}
where we have used~\eqref{eq:Kbound} and the fact that $\Phi \in D_{0,\sS}$ (which gives the constant $B_1 > 0$). It is now easy to see that the right hand side is a square integrable function of $(\bp_1,\dots,\bp_N)$ if $\a > 3/2$, $\b > 3$, $\g > 15/2$ and therefore we get 
\begin{equation*}
\big\| \wick{\de^lc^{-,+}\de^k c^{+,-}}(\d_\l K^{l,k}_{2n,\l m}*f)^{(N)}\Phi\big\| \leq B_2 C_1^{n-1}\|\hat \psi\|_\infty^n\|\hat \varphi\|_\a^n,
\end{equation*}
where $B_2>0$ is a constant depending on $m$, $f$, $\Phi$ but \emph{not} on $n$ and $\l$. A similar estimate holds then for $\| \wick{\de^lc^{-,-}\de^k c^{+,+}}(\d_\l K^{l,k}_{2n,\l m}*f)^{(N)}\Phi\|$. Furthermore we have
\begin{multline*}
\big\| \wick{\de^lc^{-,-}\de^k c^{+,-}}(\d_\l K^{l,k}_{2n,\l m}*f)^{(N-2)}\Phi\big\| \leq \\
16\pi^5\sqrt{N(N-1)}\big\| \Phi^{l,k,-}_{(\d_\l K^{l,k}_{2n,\l m})\sphat \hat f_\d}\big\|_{L^2(\bR^6)}\|\Phi\|,
\end{multline*}
with
\begin{multline*}
\big\| \Phi^{l,k,-}_{(\d_\l K^{l,k}_{2n,\l m})\sphat \hat f_\d}\big\|^2_{L^2(\bR^6)}\leq\\
\frac{C_1^{2(n-1)}B_3}{16\pi^4}\|\hat \psi\|_\infty^{2n}\|\hat \varphi\|_\a^{2n}\|\hat f\|^2_\b
\int_{\bR^6}d\bp d\bq\frac{(1+|\bp|)^3(1+|\bq|)^3}{(1+|\bp|+|\bq|)^{2\b}},
\end{multline*}
for some $\b > 6$, and, again, a similar estimate holds for $\| \wick{\de^lc^{-,+}\de^k c^{+,+}}(\d_\l K^{l,k}_{2n,\l m}*f)^{(N+2)}\Phi\| $. In summary, we get, uniformly for $\l \in [0,1/m]$,
\begin{equation*}
\big\|\wick{\de^l\phi\de^k\phi^\dagger}\big(\d_\l K^{l,k}_{2n,\l m}*f\big)\Phi\big\| \leq B_4 C_1^{n-1}\|\hat \psi\|_\infty^n\|\hat \varphi\|_\a^n,
\end{equation*}
with $B_4$ independent of $\l$ and $n$, so that, if $l+k\geq 1$, it is possible to interchange the limit and the sum in~\eqref{eq:limitseries}. The term in~\eqref{eq:limitseries} with $l=k=0$ needs however a separate treatment, due to the divergent prefactor $\l^{-1}$. We first observe that, due to the first relation in~\eqref{eq:Ksymm}, we have $\hat K^{0,0}_{n,m}(0,0) = 0$. Using bounds~\eqref{eq:Kboundderk}, \eqref{eq:Kboundderm}, we thus obtain the estimate
\begin{equation*}\begin{split}
\left|\frac{1}{\l}\hat K^{0,0}_{2n,\l m}(\l p_\s, \l q_{\s'}) \right|&= \left|\frac{1}{\l}\int_0^\l d\mu \left.\frac{d}{d\l}\hat K^{0,0}_{2n,\l m}(\l p_\s, \l q_{\s'})\right|_{\l = \mu}\right|\\
&\leq \frac{3C_2^{n-1}}{4\pi^2}\|\hat\psi\|_{1,\infty}\|\hat\varphi\|_{1,\a}(m+|\bp|+|\bq|)(1+|\bp|)^2(1+|\bq|)^2,
\end{split}\end{equation*}
valid for $\s,\s'=\pm$ and for $\l \in [0,\l_0]$, with $\l_0 := \min\{1/em,1\}$.
Then a straightforward adaptation of the above arguments easily gives, uniformly for $\l \in [0,\l_0]$,
\begin{equation}\label{eq:unifboundtaylor}
\frac{1}{\l}\big\|\wick{\phi\phi^\dagger}\big(\d_\l K^{l,k}_{2n,\l m}*f\big)\Phi\big\| \leq B_5 C_2^{n-1}\|\hat \psi\|_{1,\infty}^n\|\hat \varphi\|_{1,\a}^n,
\end{equation}
with $B_5>0$ a constant independent of $\l$ and $n$.

The same estimates above, being uniform in $\l \in [0,1/m]$,  together with use of lemma~\ref{lem:Kprop}\ref{it:cont}, allow us also to conclude that
\begin{equation}\label{eq:limitlk}
\lim_{\l \to 0}\wick{\de^l\phi\de^k\phi^\dagger}\big(\d_\l K^{l,k}_{2n,\l m}*f\big)\Phi = (2\pi)^4 \hat K^{l,k}_{2n,0}(0,0)\wick{\de^l\phi\de^k\phi^\dagger}\big(f_\d\big)\Phi.
\end{equation}
Furthermore there holds
\begin{equation*}
\lim_{\l \to 0} \left(\frac{1}{\l}\d_\l \hat K^{0,0}_{2n,\l m}* f\right)\sphat(p,q) = (2\pi)^4(p_0-q_0) \frac{\de \hat K^{0,0}_{2n,0}}{\de p_0}(0,0) \hat f(p+q),
\end{equation*}
since, as a consequence of~\eqref{eq:Ksymm}, we have $\frac{\de \hat K^{0,0}_{2n,0}}{\de p_i}(0,0) = 0 = \frac{\de \hat K^{0,0}_{2n,0}}{\de q_i}(0,0)$, $i=1,2,3$, and $\frac{\de \hat K^{0,0}_{2n,0}}{\de p_0}(0,0) = -\frac{\de \hat K^{0,0}_{2n,0}}{\de q_0}(0,0)$. Exploiting again the uniformity in $\l \in [0,\l_0]$ of the estimates leading to~\eqref{eq:unifboundtaylor}, we finally get
\begin{equation*}
\lim_{\l \to 0}\frac{1}{\l}\wick{\phi\phi^\dagger}\big(\d_\l K^{0,0}_{2n,\l m}*f\big)\Phi = -(2\pi)^4 i \frac{\de \hat K^{0,0}_{2n,0}}{\de p_0}(0,0) \wick{\de \phi \phi^\dagger-\phi\de\phi^\dagger}(f_\d).
\end{equation*}
Together with~\eqref{eq:limitlk}, this gives the statement.
\end{proof}

We stress that vanishing of the constant $c$ in the previous theorem is still by no means ruled out. That in general this is not the case, can be seen by choosing the time-smearing function $\psi \in \sD_\bR((-\t,\t))$ sufficiently close to a $\d$ function and the space-smearing function $\varphi \in \sD_\bR(B_{r+\d-\t})$ to a characteristic function. 

\begin{proposition}
Assume that the time-smearing function $\psi$ used in the construction of $\XilO(Q)$ satisfies $\psi(t) = \t^{-1}\psi_1(\t^{-1}t)$, where $\psi_1 \in \sD_\bR((-1,1))$ is such that $\int_\bR \psi_1 = 1$, and that the space-smearing function $\varphi$ is such that $\varphi \in \sD_\bR(B_{r+\d/2+\varepsilon})$, $0\leq \varphi \leq 1$ and $\varphi(\bx) = 1$ for all $\bx \in B_{r+\d/2 -\varepsilon}$, with $\varepsilon < \d/2 -\t$. Then, denoting with $c(\t,\varepsilon)$ the corresponding constant given by equation~\eqref{eq:cseries}, there holds
\begin{equation}\label{eq:climit}
\lim_{\varepsilon\to 0}\lim_{\t \to 0} c(\t,\varepsilon) = \frac{4}{3}\pi\left(r+\frac{\d}{2}\right)^3.
\end{equation}
\end{proposition}

\begin{proof}
By induction, it is straightforward to prove the following formula for $\hat K^{l,k}_{n,0}$:
\begin{multline*}
\hat K^{l,k}_{n,0}(p,q) = \frac{(-1)^{k+n-1}i^{n-l-k}\eta_n}{(2\pi)^{n+1}} \sum_{r_1,\dots,r_{n-2}}^{0,1} \sum_{\s_1,\dots,\s_{n-1}}\prod_{j=1}^{n-1}\s_j^{r_j-r_{j-1}}\times \\
\int_{\bR^{3(n-1)}} \prod_{j=1}^{n-1}d\bk_j |\bk_j|^{r_j-r_{j-1}} \hat h(p-k_{1,\s_1}) \hat h(k_{1,\s_1}-k_{2,\s_2})\dots \hat h(k_{n-1,\s_{n-1}}+q),
\end{multline*}
where $\eta_n = i$ for $n$ even and $\eta_n=-1$ for $n$ odd and $r_0 := l, r_{n-1}:= 1-k$. Since $\hat \psi(p_0) = \hat \psi_1(\t p_0) \to (2\pi)^{-1/2}$ as $\t \to 0$ and $k_{j,\s_j} = (\s_j|\bk_j|, \bk_j)$, it is easy to see that in the limit $\t \to 0$ the dependence on the $\s_j$'s drops off the integral in the second line of the above equation and therefore
\begin{equation*}
\lim_{\t \to 0} \hat K^{0,1}_{2n,0}(0,0) = \frac{(-1)^n 2^{2n-1}}{(2\pi)^{3n+1}} \hat\varphi * \dots*\hat \varphi(\boldsymbol{0}) = \frac{(-1)^n 4^n}{2(2\pi)^4} \int_{\bR^3}d\bx\,\varphi(\bx)^{2n}. 
\end{equation*}
Analogously, since $\hat \psi'(p_0) = \t \hat \psi_1(\t p_0) \to 0$ as $\t \to 0$, one has from the above formula 
\begin{equation*}
\lim_{\t \to 0}\frac{\de \hat K^{0,0}_{2n,0}}{\de p_0}(0,0) = 0.
\end{equation*}
But, thanks to the estimates~\eqref{eq:Kbound}, \eqref{eq:Kboundderk}, the convergence of the series~\eqref{eq:cseries} is uniform in $\tau$, so that one has
\begin{equation*}
\lim_{\t \to 0} c(\t,\varepsilon) = - \frac{1}{2}\sum_{n=1}^{+\infty}\frac{(-1)^n\pi^{2n}}{(2n)!}\int_{\bR^3}d\bx\,\varphi(\bx)^{2n}.
\end{equation*}
Since $\varphi$ is bounded above by the characteristic function of the ball $B_{r+\d}$ for $\varepsilon < \d/2$, the convergence of the above series is also uniform in $\varepsilon$ so that, taking into account that $\varphi$ converges to the characteristic function of the ball $B_{r+\d/2}$ when $\varepsilon \to 0$, we finally get~\eqref{eq:climit}.
\end{proof}

It is straightforward to extend the above analysis to treat the case of the net $O \to \sF(O)$ generated by a multiplet of (real or complex) free scalar fields $\phi_a$, $a = 1,\dots,d$, with the action of a compact Lie group $G$ defined by
\begin{equation*}
V(g)\phi_a(f) V(g)^* = \sum_{b=1}^dv(g)_{ab}\phi_b(f), \qquad g \in G,
\end{equation*}
where $v$ is a $d$-dimensional unitary representation (real or not  depending on whether the fields $\phi_a$ are real or complex).

More precisely, consider the one parameter subgroup $s \in \bR \to g_\xi(s) \in G$ associated to a Lie algebra element $\xi \in \mathfrak{g}$ and correspondingly the global generator $Q^\xi$ of $s \to V(g_\xi(s))$, which satisfies
on $D(N)$
\begin{equation*}
[Q^\xi,\phi_a(f)] = -i\sum_{b=1}^d t(\xi)_{ab}\phi_b(f),
\end{equation*}
$\xi \to t(\xi)$ being the representation of $\mathfrak{g}$ (through antihermitian matrices) associated to $v$. Then considering again the U(1) symmetry of the doubled theory and the associated Noether current $J_0$ it is possible to define a semi-local implementation of the flip as in equation~\eqref{eq:W} and to construct a local implementation $\XilO(Q^\xi)$ of $Q^\xi$ as in equation~\eqref{eq:Psil}, which is essentially self-adjoint on $D(N)$ and for which an expansion analogous to~\eqref{eq:Psiseries} holds:
\begin{equation*}
\XilO(Q^\xi)\Phi = \sum_{n=1}^{+\infty}\frac{\pi^{2n}}{4^n(2n)!} \bigg[\sum_{l,k}^{0,1}\sum_{a,b=1}^d t(\xi)_{ab}\wick{\de^l\phi_a\de^k\phi_b^\dagger}\big(K^{l,k}_{2n,m}(\l)\big)\Phi
\bigg],
\end{equation*}
where $K^{l,k}_{2n,m}(\l)$ are the distributions defined in~\eqref{eq:Krecursinit}, \eqref{eq:Krecurs}. Finally, the analogue of formula~\eqref{eq:limpsi} holds, where on the right hand side the appropriate Noether current
\begin{equation*}
j_0^\xi(f)= \sum_{a,b=1}^d t(\xi)_{ab}\wick{\phi_a\de\phi_b^\dagger - \de\phi_a\phi_b^\dagger}(f_\delta)
\end{equation*}
appears and the normalization constant $c$ is again given by~\eqref{eq:cseries}. 

\section{Summary and outlook}
In the present work we have shown that it is in principle possible to construct operators implementing locally a given infinitesimal symmmetry of a local net of von Neumann algebras (local generators), starting from the existence of unitary operators implementing (semi-)locally the flip automorphism on the tensor product of the net with itself.

In particular, in a large class of free scalar field models our construction provides an efficient tool to obtain manageable such local generators through the explicit expression of the local flip given in eq.~\eqref{eq:W}. Moreover, we showed that it is possible to recover, up to a well determined strictly positive normalization constant, the associated Noether currents through a natural scaling limit of these generators in which the localization region shrinks to a point. As expected, the above mentioned constant is found to depend only on the volume of the initial localization region of the generator and not on the mass and isospin of the model. The existence of this limit depends in this case on control of the energy behaviour of the generators (namely the existence of $H$-bounds) rather than on dilation invariance of the (thus massless) theory, which was a key ingredient of previous similar results~\cite{Carpi1999, Tomassini1999}.

These results have been obtained in the spirit of giving a consistency check towards a full quantum Noether theorem according to the program set down in~\cite{Doplicher:1983} and recalled in the introduction. In order to proceed further in this direction it is apparent that two main problems have to be tackled. First, it is necessary to extend the construction of local generators proposed in the introduction  to a suitably general class of theories.  Second, it would be desirable to gain a deeper understanding of the general properties granting the existence and non-triviality of the pointlike limit of the free generators, which are presently under investigation. Among other things, this is likely connected with the problem of clarifying if it is generally possible, through a suitable choice of the local flip implementation, to gain control over the ``boundary part'' of the local symmetry implementation, whose arbitrariness is considered to be an important obstruction for the reconstruction of Noether currents. The methods of~\cite{Bostelmann:2004mp} can be expected to be useful to put this analysis in a more general framework. 

Finally, we believe that our method could help to shed some light on the difficult problem of obtaining sharply localized charges from global ones.

\renewcommand{\thesection}{\Alph{section}}
\setcounter{section}{0}

\section{Local implementation of the doubled theory U(1) symmetry}\label{app:current}
In this appendix, we show that the smeared Noether current associated to the U(1) symmetry of the theory of two complex free scalar fields of mass $m\geq 0$, equation~\eqref{eq:doublesymmetry}, is represented by a self-adjoint operator which generates a group locally implementing the symmetry. Although this material is more or less standard, we include it here both for the convenience of the reader and because the proof of self-adjointness of (Wick-ordered) bilinear expressions in the free field (and its derivatives) can be found in the literature only for mass $m > 0$ (see~\cite{Langerholc:1965a, Albeverio:2004a}). For this reason, we will only emphasize the main differences in the (possibly) massless case. 

To begin with, the main estimates in the appendix of~\cite{Langerholc:1965a}, which are valid only for $m > 0$, have to be sharpened as in the following lemma.

\begin{lemma}\label{lem:bound}
Let $h \in \sS(\bR^4)$, and consider the tempered distribution $h_\d(x,y) = h(x)\d(x-y)$. Then $h_\d \in \sC^{0,1}\cap\sC^{1,0}$ for all $m\geq 0$, and $\|\hat{h}_\d\|_{0,1},\|\hat{h}_\d\|_{1,0}\leq\norm{h}_\sS$ where $\norm{\cdot}_\sS$ is some Schwartz norm independent of $m$ varying in bounded intervals.
\end{lemma}

\begin{proof}
One has $\hat{h}_\d(p,q) = \frac{1}{(2\pi)^2}\hat{h}(p+q)$, which implies $\hat{h}_\d \in \hat{\sC}$. We denote by $w(\bp,\bq)$ the integral kernel defining $T^{1,0}(\abs{\hat{h}_\d})$. It is easy to see that, for $\abs{\bq}\geq 1$,
\begin{equation*}
\sqrt{\frac{\omp}{\omq}}\leq (1+\abs{\bp-\bq})^{1/2},
\end{equation*}
and therefore, being $\hat{h} \in \sS(\bR^4)$, there exists a $C_1 > 0$ and an $r > 3$ such that
\begin{equation*}\begin{split}
\int_{\bR^3}d\bp\,\biggabs{\int_{\abs{\bq}>1}d\bq\,w(\bp,\bq)\Phi(\bq)}^2 &\leq \int_{\bR^3}d\bp\,\bigg(\int_{\abs{\bq}>1}d\bq\,\frac{C_1}{(1+\abs{\bp-\bq})^r}\abs{\Phi(\bq)}\bigg)^2 \\
&\leq C_1^2\bigg(\int_{\bR^3}\frac{d\bp}{(1+\abs{\bp})^r}\bigg)^2 \norm{\Phi}_2^2,
\end{split}\end{equation*}
where use was made of the Young inequality $\norm{f*g}_2 \leq \norm{f}_1\norm{g}_2$. On the other hand, there exist $C_2 > 0$ and $s>2$ such that, for $\abs{\bp}>1$,
\begin{equation*}\begin{split}
\biggabs{\int_{\abs{\bq}\leq 1}d\bq\,w(\bp,\bq)\Phi(\bq)} &\leq \int_{\abs{\bq}\leq 1}d\bq\,\sqrt{\frac{\omp}{\abs{\bq}}}\frac{C_2}{(1+\abs{\bp-\bq})^s}\abs{\Phi(\bq)} \\
&\leq C_2 \frac{\sqrt{\omp}}{\abs{\bp}^s}\int_{\abs{\bq}\leq 1}\frac{d\bq}{\sqrt{\abs{\bq}}}\abs{\Phi(\bq)}\leq C_2 \frac{\sqrt{2\pi\omp}}{\abs{\bp}^s}\norm{\Phi}_2,
\end{split}\end{equation*}
and a $C_3 >0$ such that, for $\abs{\bp}\leq 1$,
\begin{equation*}
\biggabs{\int_{\abs{\bq}\leq 1}d\bq\,w(\bp,\bq)\Phi(\bq)} \leq C_3 \int_{\abs{\bq}\leq 1}d\bq\,\sqrt{\frac{\omp}{\abs{\bq}}}\abs{\Phi(\bq)} \leq \sqrt{2\pi}C_3(1+m^2)^{1/4}\norm{\Phi}_2.
\end{equation*}
Putting these inequalities together, we obtain
\begin{equation*}\begin{split}
\norm{T^{1,0}(\abs{\hat{h}_\d})\Phi}_{L^2(\bR^3)} \leq  \sqrt{2}\bigg[& C_1\int_{\bR^3}\frac{d\bp}{(1+\abs{\bp})^r} + \sqrt{2\pi}C_2\Big(\int_{\abs{\bp}>1}\frac{\omp}{\abs{\bp}^{2s}}\Big)^{1/2} +\\
&+ 2\pi\sqrt{2/3}C_3(1+m^2)^{1/4}\bigg]\norm{\Phi}_{L^2(\bR^3)},
\end{split}\end{equation*}
so that, since the constants $C_i$ can be expressed by Schwartz norms of $h$, we conclude that $\norm{T^{1,0}(\abs{\hat{h}_\d})} \leq \norm{h}_\sS$ for a suitable Schwartz norm $\norm{\cdot}_\sS$.

The proofs that $\norm{T^{1,0}(\abs{\widetilde{\hat{h}}_\d})}$, $\norm{T^{0,1}(\abs{\hat{h}_\d})}$, $\norm{T^{0,1}(\abs{\widetilde{\hat{h}}_\d})}\leq\|h\|_\sS$ are completely analogous and it is immediate to see that $\Phi^{0,1,\s}_{\hat{h}_\d}, \Phi^{1,0,\s}_{\hat{h}_\d} \in L^2(\bR^6)$, $\s = \pm$, and that their norms can be bounded by $\|h\|_\sS$.
\end{proof}

This lemma, together with proposition~\ref{prop:bilin}, shows that the timelike component $J_0(h)$ of the current~\eqref{eq:currentdouble} is well-defined for $h\in\sS(\bR^4)$. Using the fact that $\abs{p_i}\leq \omp$, the proof above shows that the spacelike components $J_i(h)$, $i=1,2,3$, are well-defined too.

\begin{proposition}\label{prop:current}
The following statements hold.
\begin{proplist}{3}
\item For each $h \in \sS(\bR^4)$, the operator $J_\m(h)$ defined on $D(\tilde{N})$ by
\begin{equation}\label{eq:currentf}
J_\m(h) := \sum_{j=1}^2 (-1)^j[\wick{\de_\m\phi_j\phi^\dagger_{j'}}(h_\d) - \wick{\phi_j\de_\m\phi^\dagger_{j'}}(h_\d)],
\end{equation}
where $j' = 3-j$, defines a Wightman field such that $J_\m(h)$ is essentially self-adjoint for real $h$.
\item If $h \in \sD_\bR(O)$, $O$ a double cone, then $e^{i\zeta J_\m(h)} \in \tsF(O)$.
\item Given a 3-dimensional open ball $B_r$ of radius $r$ centered at the origin together with functions $\varphi \in \sD_\bR(\bR^3)$, $\psi \in \sD_\bR((-\tau,\tau))$ such that $\varphi(\bx) = 1$ for each $\bx \in B_{r+\tau}$ and $\int_\bR \psi = 1$, it holds that
\begin{equation}\label{eq:localimplem}
e^{i\zeta J_0(\psi\otimes \varphi)}F e^{-i\zeta J_0(\psi\otimes \varphi)} = \g_\zeta(F),\quad \forall \,F \in \tsF(O_r),
\end{equation}
where $O_r$ is the double cone with base $B_r$.
\end{proplist}
\end{proposition}

\begin{proof}
(i)  According to lemma~\ref{lem:bound} one has $\|\hat{h}_\d\|_{0,1},\|\hat{h}_\d\|_{1,0}\leq\norm{h}_\sS$, so that $J_\m$ is a Wightman field and $J_\m(h)$ is symmetric for real $h$. Consider now a $\Phi \in K^{\otimes_S n}$. $J_\m(h)^p\Phi$ is the sum of $16^p$ vectors of the form $\wick{\de_\m^{l_p}c^{-,\s_p}_{j_p}\de_\m^{k_p}c^{+,\eps_p}_{j'_p}}(h)^{(n_p)}\dots\wick{\de_\m^{l_1}c^{-,\s_1}_{j_1}\de_\m^{k_1}c^{+,\eps_1}_{j'_1}}(h)^{(n_1)}\Phi$ with $n_j = n_{j-1} + \s_j+\eps_j$, $j=1,\dots,p$ ($n_0 := n$). Therefore, by~\eqref{eq:boundbilc},
\begin{equation*}
\norm{J_\m(h)^p\Phi}\leq \left(\frac{4\norm{h}_{\sS}}{\pi}\right)^p (n+2(p+1))\dots(n+4)\norm{\Phi},
\end{equation*}
and we see that $\Phi$ is an analytic vector for $J_\m(h)$. Since any element in $\tilde{D}_0$ is a finite sum of such vectors, essential self-adjointness of $J_\m(h)$ follows. 

(ii) A straightforward but lengthy calculation shows that, on $\tilde{D}_0$,
\begin{equation}\begin{split}\label{eq:commut}
[J_\m(h),\phi_j(f)&+\phi_j(f)^*] = (-1)^{j+1}i(\phi_{j'}(g)+ \phi_{j'}(g)^*),\\
g &=h (\de_\m\D*f)+\de_\m(h (\D*f)),
\end{split}
\end{equation}
where, as costumary, $\D$ is the Fourier transform of $\frac{1}{2\pi i}\eps(p_0)\d(p^2-m^2)$. 
Since $\supp \D$ is contained in the closed light cone and $\tilde{D}_0$ is an invariant dense set of analytic vectors for both $J_\m(h)$ and $\phi_j(f)+\phi_j(f)^*$, we see by standard arguments that $e^{i\zeta J_\m(h)}$ commutes with $e^{i[\phi_j(f)+\phi_j(f)^*]^-}$ if $\supp h$ is spacelike from $\supp f$, i.e. $e^{i\zeta J_\m(h)} \in \tsF(O)'' = \tsF(O)$ if $\supp h \subset O$.

(iii) Take $f \in \sD(O_r)$. Since $\supp \D *f$ does not intersect $[-\t,\t]\times\{\bx \,:\,\varphi(\bx)\neq 1\}$ we have that
\begin{equation*}
 \psi\otimes \varphi (\de_0\D*f)+\de_0(\psi \otimes \varphi (\D*f)) = \psi\otimes 1 (\de_0\D*f)+\de_0(\psi\otimes 1 (\D*f)).
\end{equation*}
On the other hand a calculation shows that, thanks to $\int_\bR \psi = 1$,
\begin{equation*}
\D*\big(\psi\otimes 1 (\de_0\D*f)+\de_0(\psi\otimes 1 (\D*f))\big) = \D*f,
\end{equation*}
and, since $\D * f_1 = 0$ implies $f_1 = (\Box + m^2)f_2$  with $f_i \in \sS(\bR^4)$, the commutation relations \eqref{eq:commut}
become
\begin{equation*}
[J_0(\psi\otimes\varphi),\phi_j(f)+\phi_j(f)^*] = (-1)^{j+1}i(\phi_{j'}(f)+ \phi_{j'}(f)^*).
\end{equation*}
Furthermore, 
thanks to the estimates~\eqref{eq:Nbound}, \eqref{eq:Nboundcomm} we can apply the multiple commutator theorems in~\cite{Frohlich:1977yz} to conclude, as in the proof of theorem 2 in~\cite{D'Antoni:1983a}, that~\eqref{eq:localimplem} holds.
\end{proof}

\emph{Acknowledgements.} We would like to thank Sergio Doplicher for originally suggesting the problem to one of us and for his constant support and encouragement, and Sebastiano Carpi for several interesting and useful discussions. We also thank the referees for suggesting several improvements to the exposition.

\providecommand{\bysame}{\leavevmode\hbox to3em{\hrulefill}\thinspace}

\end{document}